\documentclass[reqno]{amsart}

\usepackage{a4wide,amsrefs,amssymb}
\usepackage[colorlinks=false]{hyperref}
\usepackage{esint}

\title[Weak second Bianchi identity]{Weak second Bianchi identity\\ for static,
spherically symmetric spacetimes\\ with timelike singularities}

\author[A.~Burtscher]{Annegret Burtscher} 
\address{Department of Mathematics, Institute for Mathematics, Astrophysics and
Particle Physics (IMAPP), Radboud University, PO Box 9010, 6500 GL Nijmegen, The
Netherlands}
\email{\href{mailto:burtscher@math.ru.nl}{burtscher@math.ru.nl}}

\author[M.~Kiessling]{Michael K.-H.\ Kiessling} 
\address{Department of Mathematics, Rutgers, The State University of New Jersey,
110 Frelinghuysen Rd., Piscataway, NJ 08854-8019, United States of America}
\email{\href{mailto:miki@math.rutgers.edu}{miki@math.rutgers.edu}}

\author[A.S.~Tahvildar-Zadeh]{A.\ Shadi Tahvildar-Zadeh} 
\email{\href{mailto:shadi@math.rutgers.edu}{shadi@math.rutgers.edu}}

\makeatletter
\@namedef{subjclassname@2020}{\textup{2020} Mathematics Subject Classification}
\makeatother

\subjclass[2020]{83C50 (primary), 35Q75, 83C22, 70H40, 53B30, 53C80 (secondary)}

\keywords{second Bianchi identity, general relativity, energy-momentum conservation,
naked singularities, particles, electromagnetism}

\thanks{\textit{Acknowledgements}: We thank John Stalker for valuable discussions
following the completion of the first version of this manuscript. We are also
grateful to the anonymous referees for their constructive criticism.
AB acknowledges support of her research by the Dutch Research Council (NWO) grant number 
VI.Veni.192.208. Part of this material is
based upon work supported by the Swedish Research Council under grant no.\ 2016-06596
while AB was in residence at Institut Mittag-Leffler in Djursholm, Sweden during the
research program ``General Relativity, Geometry and Analysis: Beyond the first 100 
years after Einstein'' in Fall 2019.}

\date{\today}

\numberwithin{equation}{section}

\theoremstyle{plain}
\newtheorem{theorem}{Theorem}[section]
\newtheorem{corollary}[theorem]{Corollary}
\newtheorem{proposition}[theorem]{Proposition}

\newtheorem{innercustomthm}{Theorem}
\newenvironment{customthm}[1]
  {\renewcommand\theinnercustomthm{#1}\innercustomthm}
  {\endinnercustomthm}

\theoremstyle{definition}
\newtheorem{definition}[theorem]{Definition}
\newtheorem{example}[theorem]{Example}

\theoremstyle{remark}
\newtheorem{remark}[theorem]{Remark}

\def\RR{\mathbb{R}}
\def\SS{\mathbb{S}}

\newcommand{\g}{\ensuremath{\mathbf{g}}}
\newcommand{\F}{\ensuremath{\mathbf{F}}} 
\newcommand{\M}{\ensuremath{\mathbf{M}}} 
\newcommand{\T}{\ensuremath{\mathbf{T}}} 
\newcommand{\R}{\ensuremath{\mathbf{R}}} 
\newcommand{\bL}{\ensuremath{\mathbf{L}}} 
\newcommand{\A}{\ensuremath{\mathbf{A}}} 

\newcommand{\vol}{\operatorname{vol}}

\newcommand{\ddiv}{\operatorname{div}}

\newcommand{\mEL}{{m}_{\text{\textrm{e}}}}       
\newcommand{\eEL}{{e}_{\text{\textrm{e}}}}       

\newcommand{\Sset}{{\mathbb S}}
\newcommand{\bn}{\mathbf{n}}

\newcommand{\p}{\partial}
\newcommand{\cK}{{\mathcal K}}
\newcommand{\cM}{{\mathcal M}}
\newcommand{\cS}{{\mathcal S}}
\newcommand{\cT}{\mathcal{T}}
\newcommand{\al}{\alpha}
\newcommand{\ga}{\gamma}
\newcommand{\Ga}{\Gamma}
\newcommand{\Si}{\Sigma}
\newcommand{\bsi}{\boldsymbol{\sigma}}
\newcommand{\nab}{\nabla}
\newcommand{\half}{\frac{1}{2}}

\newcommand{\RWN}{{Reissner--Weyl--Nordstr\"om}}
\newcommand{\ppp}{{\prime\prime\prime}}
\newcommand{\pp}{{\prime\prime}}

\newcommand{\beq}{\begin{equation}}
\newcommand{\eeq}{\end{equation}}

\newcommand{\bna}{\begin{eqnarray}}
\newcommand{\ena}{\end{eqnarray}}
\newcommand{\bea}{\begin{eqnarray*}}
\newcommand{\eea}{\end{eqnarray*}}

\begin{document}


\begin{abstract}
The (twice-contracted) second Bianchi identity is a differential curvature identity that 
holds on any smooth manifold with a metric. In the case when such a metric is Lorentzian 
and solves Einstein's equations with an (in this case inevitably smooth) 
energy-momentum-stress tensor of a ``matter field'' as the source of spacetime curvature, 
this identity implies the physical laws of energy and momentum conservation for the 
``matter field''. The present work inquires into whether such a Bianchi identity can still
hold in a weak sense for spacetimes with curvature singularities associated with timelike
singularities in the ``matter field''. Sufficient conditions that establish a distributional 
version of the twice-contracted second Bianchi identity are found.
In our main theorem, a large class of spherically symmetric static
Lorentzian metrics with timelike one-dimensional singularities is identified, for
which this identity holds. As an important first application we show that the 
well-known Reissner--Weyl--Nordstr\"om spacetime of a point charge does not belong 
to this class, but that Hoffmann's spacetime of a point charge with negative bare mass in
the Born--Infeld electromagnetic vacuum does. 
\end{abstract}

\maketitle

\vspace{-1truecm}
\section{Introduction and Main Results}
\label{intro} 

\subsection{Motivation}

Einstein's equations for the spacetime metric\footnote{The signature of a Lorentzian
metric $\g$ is $(-,+,+,+)$. Greek indices $\mu,\nu$ etc.\ denote the components
$0,1,2,3$ of a tensor defined on the spacetime, with respect to a local coordinate
system $(y^\mu)_{\mu=0}^3$; however, Cartesian coordinates are denoted 
$(x^\mu)_{\mu=0}^3$. The coordinate vector fields are written 
$\partial_\mu=\frac{\partial}{\partial y^\mu}$. We use the Einstein summation
convention. To facilitate discerning the physical meaning of our results, we retain
$G$ and $c$.} $\g = (g_{\mu\nu})$ of a $3+1$-dimensional Lorentzian manifold 
$(\cM,\g)$ read
\beq\label{EM1}
 \R - \tfrac12 R \g = \tfrac{8\pi G}{c^4} \T[\F];
\eeq
here, $\R=(R_{\mu\nu})$ denotes the Ricci curvature tensor of the metric $\g$, 
$R = g^{\mu\nu}R_{\mu\nu}$ is its scalar curvature, $G$ is Newton's constant of
universal gravitation, and $c$ is the speed of light in vacuum. Moreover, 
$\T=(T_{\mu\nu})$ is the {\em energy-momentum-stress tensor} of any ``matter field''
$\F$ in (or associated with) the spacetime. By using this hybrid terminology of 
``matter field'' to cover models of continuum fluids, elastic solids, etc., as well 
as the electromagnetic field, we follow common practice in the general relativity 
community; from now on, we will drop the scare quotes. Einstein's equations~\eqref{EM1} 
often need to be complemented by evolution equations for the matter field $\F$.

For any sufficiently regular Lorentzian metric (classically, $\g\in C^3$), the
(twice-contracted) {\em second Bianchi identity}
\beq\label{Bi2}
\nab_\mu \left(R^\mu{}_\nu  - \tfrac12 R g^\mu{}_\nu \right) = 0
\eeq
holds; here, $\nab$  denotes the {Levi-Civita connection} associated with $\g$. As a
consequence, for any solution $(\g,\F)$ of \eqref{EM1} which is regular enough so
that this equation, as well as the equation obtained by differentiating both sides of
\eqref{EM1}, is satisfied pointwise, the identity \eqref{Bi2} implies the matter
field's local conservation laws of energy-momentum 
\beq\label{eq:conservationlaws}
\nab_\mu T^\mu{}_\nu = 0.
\eeq
Matter field equations must be compatible with \eqref{eq:conservationlaws},
therefore.

If $\F$ represents a perfect fluid with barotropic equation of state, then for
sufficiently regular evolutions (e.g., prior to any shock formation) the space
component of \eqref{eq:conservationlaws} is part of the equations of the fluid
evolution, to be complemented merely by the continuity equation for the fluid; the
time component of \eqref{eq:conservationlaws} is then redundant. On the other hand,
if $\F$ represents a source-free electromagnetic field, then
\eqref{eq:conservationlaws} does not furnish field evolution equations; they need to
be stated separately, compatible with \eqref{eq:conservationlaws}.

In a series of influential papers, \cites{EIH38,EI40,EI49}, Einstein and Infeld (EI),
originally joined by Hoffmann (EIH), claimed that the field equations of general
relativity theory, \eqref{EM1}, coupled with the Maxwell--Lorentz evolution equations
for the electromagnetic fields, determine the equations of motion of matter modeled
atomistically as composed of charged point particles, which they identified with
point singularities in spacelike slices of a spacetime. They actually discussed
mostly the special case of uncharged particles; a follow-up paper by Infeld's student
Wallace \cite{Wal41} supplied more details about their claim concerning charged point
singularities.

The idea that the world lines of point particles should be replaced by one-dimensional
timelike singularities of spacetime seems to go back to Weyl \cite{Wey21}. Already in
Weyl's writing, it is clear that such singularities are not subsets of the spacetime\footnote{Some 
care is thus required when one talks about causal properties of
singularities such as them being timelike, spacelike, or null, since these notions
can only be meaningful in a limiting sense.}. Thus, in this setup the ``world line of
a particle'' is not a path \emph{in} spacetime but a timelike one-dimensional singularity 
\emph{of} spacetime --- or put differently: an interior boundary of the spacetime --- which
needs to be determined along with the spacetime. {\em If} this setup can be
consistently implemented into general relativity, it produces spacetimes with 
one-dimensional timelike curvature singularities that have the appearance of world
lines of charged point particles, which are the sources and sinks of the
electromagnetic fields living in this spacetime. This is different from the usual
textbook story of ``test particle'' motion which, when uncharged, is given by a
timelike geodesic in a spacetime that is defined independently of the particle's
existence.

Although non-rigorous and full of questionable assumptions, and with conclusions
which cannot possibly be true in the sweeping generality in which they were stated,
the EIH papers have become the template for many formal follow-up calculations (for a
survey see, e.g., \cite{PoissonETal}), in particular the computation of gravitational
wave signals and their feedback on the motion of the sources (binary neutron stars or
black holes) used for the interpretation of the LIGO and VIRGO gravitational wave
data \cite{Bla2014}. As far as we can see, existing rigorous works on the problem of
motion for ``small bodies'' in general relativity (we mention in particular
\cites{GJ,QW,EG,St,GrWa,GrWaERR,GW,Harte}) do not yet allow a definitive assessment
of the merits of some of the key ideas of the papers by Einstein, Infeld, and
Hoffmann, and by Wallace, on the motion of what nowadays would be called naked
singularities. The purpose of the present paper is to take one step further toward
this goal.

A rigorous assessment would require one to consistently formulate an at least locally
well-posed joint evolution problem for spacelike slices of a spacetime, the
electromagnetic field defined on these slices, and the point singularities (in the
spatial curvature tensor) that represent the sources and sinks of that field.
Moreover, if the EIH and EI claims have any merit, then the equations of motion for
the point singularities in the spacelike slices must be a consequence of \eqref{EM1},
coupled with Maxwell's evolution equations for the electromagnetic fields $\F$ which,
however, need to be supplemented by a suitable law of the electromagnetic vacuum. It
is clear that such an electromagnetic law must be different from the usual Maxwell
vacuum law, for the latter leads to infinite electromagnetic field energy of the
point singularity, i.e., non-integrable singularities in the electromagnetic $\T$
that cause non-integrable curvature singularities of the metric, as per Einstein's
field equations \eqref{EM1} (see Section \ref{sec:Hamiltonians} for more details). 

To get an idea of the mathematical subtleties that could be involved, suppose that 
such a spacetime {$\cM$} with timelike one-dimensional singularities can be
continuously extended (just the manifold, not  the metric) into the location of these
singularities. In such a situation these one-dimensional timelike singularities
become proper particle world lines in this extended spacetime, and it is then
meaningful to express the energy-momentum-stress tensor $\T$ as a sum of a regular
and a singular part,
\beq
\T = \T^{\mbox{\tiny reg}} + \T^{\mbox{\tiny sing}},
\eeq
with $\T^{\mbox{\tiny reg}}$ sufficiently regular away from the world lines of the
point-charges, and $\T^{\mbox{\tiny sing}}$ the usual energy-momentum-stress tensor
of a system of point particles which is supported only on these world lines in a weak
sense as a measure. Now, if \eqref{Bi2}  holds in a weak sense it then follows
that \eqref{eq:conservationlaws} must hold in a weak sense as well, and hence
\beq\label{divTrPLUSdivTsISnull}
\nab\cdot \T^{\mbox{\tiny sing}} = -\nab\cdot\T^{\mbox{\tiny reg}}
\eeq
in the sense of distributions. In \cites{Kie19,Kie20,KieTah21} it has been shown that
in special-relativistic electromagnetic spacetimes (i.e., Newton's constant $G \to 0$
in \eqref{EM1}) with suitable electromagnetic vacuum laws the total (electromagnetic)
force on a charged point singularity, and its classical equation of motion, can be
extracted from \eqref{divTrPLUSdivTsISnull}. Furthermore, for the 
Bopp--Land\'e--Thomas--Podolsky (BLTP) vacuum law, the special-relativistic joint
initial value problem for point charges coupled with the electromagnetic 
Maxwell--BLTP fields is locally well-posed in time \cite{KieTah21}. Hence it is
reasonable to expect, possibly under further conditions on the behavior of $\F$, that
at least some of the well-posed special-relativistic joint initial value problems
that can be extracted from \eqref{divTrPLUSdivTsISnull} at $G=0$, can be continuously
extended to the general-relativistic domain when $G>0$. Which matter field models 
$\F$ qualify in this sense is an important open problem; we offer some remarks in the
last section.
 
Independent of the inquiry into suitable matter field models, the following is now
clear: For the establishment of the energy-momentum conservation law
\eqref{divTrPLUSdivTsISnull} when $G>0$ that could pave the ground toward a 
well-posed joint initial value problem for the spacelike slices of spacetime, the
electromagnetic and perhaps other matter fields in it, and their charged point 
singularities, along similar lines as in the special-relativistic formulation
mentioned above, it is necessary that the second Bianchi identity \eqref{Bi2} holds
in a weak sense. Thus, the key question is:

\noindent
\emph{Under which conditions on the metric of the spacetime does the weak second
Bianchi identity
\beq
   \int_{{\cM}} \left(R^{\mu}{}_\nu-\tfrac12 R g^\mu{}_\nu \right) \, \nabla_\mu 
   \psi^\nu \, d\vol_\g = 0
 \eeq
hold for all smooth, compactly supported vector fields $\psi$ defined on the
spacetime?}

Answering this question, in all its generality, is a big challenge, because a complete
classification of singularities of solutions of Einstein's equations seems currently 
out of reach. For example, two timelike singularities in a given spacetime can be
vastly different in terms of strength, in the sense that a curvature invariant may
blow up at two very different rates for them.

Our strategy is to begin by restricting the key question to special families of
spacetimes, incrementally becoming more general. There are many explicit solutions of
Einstein's equations where the causal structure is simple enough that everything can
be worked out explicitly and the singular behavior can be fully analyzed; in
particular, we mention distributional approaches in
\cites{GT,HS,HK,LFM,L1,L2,S,SV,Taub}. Such model cases can give us clues as to what
the sufficient conditions {are} for a spacetime singularity to represent the world
line of a particle, and which type of ``atomic matter'' models can accommodate such
singularities. We are particularly interested in ``electromagnetic matter'', whose
electromagnetic field satisfies the pre-metric Maxwell's equations, complemented with
a suitable electromagnetic vacuum law, and with charged sources given by a finite
number of one-dimensional timelike singularities that are assigned an 
energy-momentum-stress tensor in the spirit of EIH, and Wallace.

\subsection{Setting}
\label{sec:setting}
 
While a number of ideas developed below are clearly adaptable to more general
situations, in the present paper we focus our efforts on {\em static spherically
symmetric} spacetimes that feature a single timelike singularity, with special
emphasis given to electrostatic spacetimes of a single point charge at the center of
symmetry. These are four-dimensional Lorentzian manifolds $(\cM,\g)$ on which there
exists a \emph{global} system of coordinates $(t,r,\vartheta,\varphi)$ such that the
line element of the metric $\g$ can be written as
\beq\label{SSSS}
ds_\g^2 = -e^{2\al(r)} c^2dt^2 + e^{2\beta(r)} dr^2 + r^2 (d\vartheta^2 + 
\sin^2\vartheta d\varphi^2).
\eeq
Thus $\p/\p t$ is a timelike Killing field, $r>0$ is the area-radius coordinate, and
$(\vartheta,\varphi)$ are spherical coordinates on the standard sphere $\Sset^2$.

It is common knowledge that many of the known solutions of Einstein's equations
\eqref{EM1} have analytical extensions that feature geometric singularities
associated with geodesic incompleteness and/or curvature blow-up. Famous examples are
the Schwarzschild solution, both in the positive mass (black hole) as well as in the
negative mass (naked singularity) sector, and the charged \RWN\ solution in the
superextremal (naked) as well as extremal and subextremal sectors (black holes). It
is in fact a theorem \cite{TZ1} that there are no static, spherically symmetric,
asymptotically \RWN\ electrovacuum spacetimes whose maximal extension is devoid of
singularities. 

The negative mass Schwarzschild (nmS) solution and, in the superextremal (naked)
sector, also the \RWN\ (RWN) spacetimes possess global coordinate systems in which
their metric has the form \eqref{SSSS}.  In these coordinates, there is a severe 
curvature singularity on the timelike line $r=0$, so that $\cM$ as a Lorentzian
manifold is diffeomorphic to $\RR^4$ minus a line, or 
$\RR\times (\RR^3\setminus\{0\})$. The fact that the singular set of these and many
other spacetimes is of codimension three creates complications for the study of the
geometry of these manifolds in the neighborhood of their singularities, including the
question of whether or not it is possible to formulate a weak version of the second
Bianchi identity for these singular spacetimes.

One approach taken by differential geometers that has been fruitful in this regard is
the use of a different coordinate system on these manifolds, one which ``blows up'' a
neighborhood of the singular set in such a way that in the new coordinates, one has a
manifold with a codimension-one boundary, with its own intrinsic smooth geometry,
thereby allowing for tools of geometric analysis to be applied to it. Examples of
this approach can be found in the Riemannian setting in the works of Bray \cite{Bray}, 
Bray--Jauregui \cite{BJ} and others. We describe this geometric approach
in Section~\ref{sec:ZAScoord} and a corresponding notion of mass for these
codimension one boundaries in Section~\ref{sec:mass}. Finally, in
Section~\ref{sec:mainresults} of this introduction we state and discuss the main
results of this paper regarding the weak second Bianchi identity and applications,
which makes use of this geometric formulation.

\subsection{Spatial conformally flat coordinates and Zero-Area Singularities (ZAS)}
\label{sec:ZAScoord}

Let $(\cM,\g)$ be a four-dimensional static, spherically symmetric Lorentzian
manifold, diffeomorphic to $\RR\times(\RR^3\setminus\{0\})$, which has a global
spherical coordinate system $(t,r,\vartheta,\varphi)$ defined on it in which the metric
$\g$ has the form \eqref{SSSS}. Suppose we can transfer to a new coordinate 
$\rho \in (\rho_0,\infty)$, $\rho_0>0$ such that 
\beq\label{eq:isotropic}
ds_\g^2 = - e^{2\ga(\rho)} {c^2} dt^2 + \phi^4(\rho)\left[ d\rho^2 + \rho^2 
(d\vartheta^2 +\sin^2\vartheta d\varphi^2)\right].
\eeq
We call the coordinate system $(t,\rho,\vartheta,\varphi)$ in \eqref{eq:isotropic} 
\emph{spatially conformally flat coordinates}. Note that the metric inside the
brackets is simply the Euclidean metric on $\RR^3$ in spherical coordinates.
Comparing to \eqref{SSSS} we see that it is necessary that
\beq
\phi^2 = \frac{r}{\rho},\qquad e^\beta dr = \phi^2 d\rho.
\eeq
Solving the differential equation, we thus require
\beq\label{def:rho}
\rho(r):= \rho_0 \exp\left(\int_0^r e^{\beta(r')} \frac{dr'}{r'} \right),
\eeq
assuming that the integral exists (we will see that it does for the manifolds of
interest to us). Clearly, $\rho(0) = \rho_0$. Moreover, $\rho$ as defined in 
\eqref{def:rho} is an increasing function of $r$, and hence invertible, which
determines $\phi$ and $\ga$ as
\beq 
\phi(\rho) := \sqrt{{r(\rho)}/{\rho}},\qquad \ga(\rho) := \al(r(\rho)).
\eeq
The manifold $\cM$ in these coordinates is diffeomorphic to the exterior of the solid
cylinder, i.e., $\RR\times(\RR^3\setminus  \overline{B_{\rho_0}(0)})$.
We note that $\phi(\rho_0)=0$ and therefore by \eqref{eq:isotropic} the interior boundary 
$\cT = \p \cM$ is a singular boundary. In particular, the intrinsic area of the sphere 
$\p B_{\rho_0}(0)$ (which can be computed by a limiting process, see \cite{Bray}) is zero. 
The constant-$t$ slices of $\cM$ are spacelike hypersurfaces diffeomorphic to the 
exterior $\Si$ of the open ball $B_{\rho_0}(0)$ in $\RR^3$.
We can write $\cM = \RR\times\Si$, and $\cS$ for the interior boundary $\p \Si$. 
The surface $\cS$ is therefore an example of a {\em Zero Area Singularity} (ZAS) for 
the Riemannian manifold $\Si$.

In his pioneering work \cite{Bray} on ZAS, Bray defined a notion of mass for such 
singularities, and studied its properties. In particular he showed that this mass is 
coordinate invariant and always {\em negative}.  In this work we will connect Bray's notion 
of the mass of a ZAS with our notion of the (negative) bare mass of the
central singularity in the spacetimes discussed in this paper.

In the rest of this section we derive sufficient conditions for obtaining a
coordinate transformation of the form~\eqref{eq:isotropic}. Let us recall that for
spherically symmetric spacetimes, with $r$ denoting the area-radius, one can define
the {\em cumulative mass function} $m(r)$ via the relation
\beq
1 - \frac{2Gm(r)}{c^2 r} = {g^{\mu\nu}}\p_\mu r \p_\nu r.
\eeq
Thus in terms of the metric coefficients in \eqref{SSSS},
\beq\label{cummass} 
m(r) := \frac{c^2}{2G} r\left(1 -  e^{-2\beta(r)}\right).
\eeq
Using the mass function we can rewrite \eqref{def:rho} as
\beq\label{eq:rhomass}
\rho(r) = \rho_0 \exp\left( \int_0^r \frac{dr'}{\sqrt{r'^2 - \frac{2G}{c^2}r'm(r')}}\right)
\eeq
Let us assume that $m: (0,\infty) \to \RR$ is a $C^1$ function, with the asymptotics
\beq\label{asymp:m}
 m(r) \sim \left\{\begin{array}{ll} m_0 + m_1 r & r \to 0,\\
M-\frac{M_1}{r} & r \to \infty,\end{array}\right.
\eeq
for constants $m_0 < 0$, $M\in\RR$, $m_1, M_1>0$.  Thus $m_0 = \lim_{r\to 0} m(r)$ 
and $M = \lim_{r\to\infty} m(r)$.  Moreover, for the denominator in 
\eqref{eq:rhomass} we have
\beq
r^2 - \frac{2G}{c^2} r m(r) \sim A r + B r^2 \mbox{ as } r\to 0,
\eeq
with
\beq
A := \frac{-2Gm_0}{c^2} >0,\qquad B := 1 - \frac{2Gm_1}{c^2},
\eeq
where $B\in(0,1)$ if $m_1$ is sufficiently small. Under these assumptions the
integral in \eqref{eq:rhomass} is seen to be finite and $\rho$ is well-defined. We
will mostly be interested in spacetimes where \eqref{asymp:m} holds, with the stated
range of the parameters.

\subsection{Mass of a regular ZAS}
\label{sec:mass}

Let $(\Si,\bsi)$ be  a 3-dimensional Riemannian manifold with boundary, and let 
$\cS_0 \subseteq \p\Si$ be a ZAS for $\Si$. According to Bray, a ZAS is called 
{\em regular} if it can be conformally extended, i.e., if there exists a smooth
nonnegative function $\bar{\phi}$ defined in a neighborhood of $\cS_0$ in $\Si$, and
a smooth metric $\bar{\bsi}$ such that
\begin{enumerate}
\item $\bar{\phi}= 0$ on $\cS_0$.
\item $\bar{\bn}(\bar{\phi})>0$ where $\bar{\bn}$ is the unit normal of $\cS_0$ with 
respect to $\bar{\bsi}$.
\item $\bsi = \bar{\phi}^4 \bar{\bsi}$.
\end{enumerate}
The prototype example of a regular ZAS is the singularity at the center of the
spacelike time-slices of the negative mass Schwarzschild (nmS) spacetime, which is
the unique static, spherically symmetric, asymptotically flat solution of vacuum
Einstein equations whose central singularity is not shielded by a horizon.  For this
manifold, $m(r) \equiv m_0 < 0$,  and it is easy to see that the coordinate $\rho$
as defined in Section~\ref{sec:ZAScoord} above is a global coordinate: if we set
\beq \rho_0 := \frac{G|m_0|}{2c^2}
\eeq
then the change of coordinate $r \leftrightarrow \rho$ is given by
\beq
\rho = \half \left( r + 2\rho_0 + \sqrt{ (r+2\rho_0)^2 - 4\rho_0^2}\right),\qquad
r = \frac{(\rho - \rho_0)^2}{\rho}.
\eeq
Moreover, the function $\bar{\phi}$ in the definition of regular ZAS exists and is
\beq
\bar{\phi}(\rho) = 1 - \frac{\rho_0}{\rho},
\eeq
while the metric $\bar{\bsi}$ is the Euclidean metric on 
$\RR^3\setminus B_{\rho_0}(0)$.

For ZAS that are regular, Bray~\cite{Bray} defined a notion of mass by 
\beq
m_\mathrm{reg}(\cS_0) := -\frac{1}{4} \left( \frac{1}{\pi}\int_{\cS_0} \bar{\bn}
(\bar{\phi})^{4/3} dS_{\bar{\bsi}}\right)^{3/2},
\eeq
where $dS_{\bar \bsi}$ denotes the surface element with respect to $\bar \bsi$. Note
that this mass is always negative. Let $(\cS_n)_n$ be a sequence of closed $C^2$
surfaces in $\Si$, each one being a graph over $\cS_0$, that shrink down to $\cS_0$
as $n\to\infty$. It can be shown \cite{Bray} that 
\beq
m_\mathrm{reg}(\cS_0) = \lim_{n\to \infty} m_\mathrm{H}(\cS_n),
\eeq
where $m_\mathrm{H}$ denotes the {\em Hawking mass}
\beq
m_\mathrm{H}(\cS) := \sqrt{\frac{|\cS|_{{\bsi}}}{16\pi}} \left( 1 - \frac{1}{16\pi}
\int_\cS H^2 dS_{\bsi} \right).
\eeq
Here $H$ denotes the mean curvature of the surface $\cS$ in $\Si$. Note that for the
Hawking mass to be well-defined the metric $\bsi$ and the surface $\cS$ need to be at
least $C^2$.

If the three-dimensional manifold $\Si$ is spherically symmetric, it is not hard to
see that the Hawking mass of the sphere $S_r$ of area-radius $r>0$ in the manifold is
equal to $m(r)$ obtained in \eqref{cummass}. It thus follows that for the spacetimes
$(\cM,\g)$ of interest in this paper, which admit a metric of the form
\eqref{eq:isotropic} such that the interior boundary $\cS_0$ of the time-slices is a
ZAS, the bare mass of the singularity is equal to its Hawking mass and to its Bray 
mass, i.e.,
\beq
m_0 = \lim_{r \to 0} m(r) = \lim_{\rho \to \rho_0} m_\mathrm{H}(S_\rho) = 
m_\mathrm{reg}(\cS_0).
\eeq

Besides the nmS spacetime that contains such a ZAS singularity, we introduce in
Section~\ref{prototype} a particular asymptotically flat, electrovacuum spacetime.
This prototype spacetime represents the vacuum outside a static point charge in a
nonlinear electromagnetic theory with an admissible reduced Hamiltonian as discussed
in some detail also in Section~\ref{sec:mainresults} below. Most importantly, we will
show that the central singularity of this electrovac spacetime is of the same
strength as the singularity of the nmS, i.e., it has finite negative bare mass, and
is therefore a regular ZAS.

\subsection{Main Results}
\label{sec:mainresults}

We are ready to state our results about the weak analogue of the twice-contracted
second Bianchi identity \eqref{Bi2}, which in a suitable weak sense also makes sense
on certain spacetimes with a single timelike singularity. In this paper, we prove the
following result.

\begin{theorem}\label{thm:bianchi}
Let $\cM\cong\RR \times (\RR^3 \setminus \overline{B_{\rho_0}(0)})$ be equipped with
a static, spherically symmetric Lorentzian metric $\g$ of the form 
\[
ds_\g^2 = - e^{2\ga(\rho)} c^2 dt^2 + \phi^4(\rho)\left[ d\rho^2 + \rho^2 (d\vartheta^2 
+\sin^2\vartheta d\varphi^2)\right],
\]
for spatially conformally flat coordinates $(t,\rho,\vartheta,\varphi)$. Assume
furthermore that, as $\rho \downarrow \rho_0$,
\begin{enumerate}
  \item $\phi(\rho) = O(\rho-\rho_0)$, $\phi'(\rho) = O(1)$, 
  \item $e^{\gamma(\rho)} = O((\rho-\rho_0)^{-1})$, $\gamma'(\rho) 
  = O((\rho-\rho_0)^{-1})$, and 
  \item $G^\mu{}_\nu(\rho) = O((\rho-\rho_0)^{-5+\kappa})$, for some $\kappa>0$.
\end{enumerate}
Then the second Bianchi identity is satisfied weakly, in the sense that
 \beq\label{bintA}
   \int_\cM G^\mu{}_\nu \nabla_\mu \psi^\nu d\vol_\g = 0,
 \eeq
for any compactly supported vector field\footnote{Note that the vector field can have
support on the inner boundary of the manifold.} 
$\psi \in \mathfrak{X}_c(\overline{\cM})$.
\end{theorem}

Theorem~\ref{thm:bianchi} is our starting point for an in-depth analysis of well-known static,
spherically symmetric solutions of the Einstein equations with singularities.
Our condition (iii) is a mild condition suggested by our method of
proof, but may seem a bit unnatural because it involves the independent metric functions 
$\gamma$ and $\phi$ and their first and second derivatives. In Corollary~\ref{thm:bianchi2} in Section~\ref{secBIANCHI} we replace condition (iii) in Theorem~\ref{thm:bianchi} by stronger 
conditions (i') on $\phi$ and its first and second derivatives, and (ii') on 
$\gamma$ and its first and second derivatives.  In Corollary~\ref{cor1} a special 
result is also obtained for the simpler case when the metric exponents satisfy $\beta=-\alpha$ in
\eqref{SSSS} (cf.\ \cite{TedJ}) and the cumulative mass function $m$ has the Taylor expansion 
assumed in \eqref{asymp:m} as $r \downarrow 0$.

Interestingly, we find that the superextremal Reissner--Weyl--Nordstr\"om solution,
i.e., the well-known spherically symmetric, asymptotically flat solution of the
Einstein--Maxwell--Maxwell\footnote{The first ``Maxwell'' here stands for the 
pre-metric Maxwell field equations, the second ``Maxwell'' for  Maxwell's
electromagnetic vacuum law [which he called ``law of the pure ether'']. In the same
vein we will speak of Einstein--Maxwell--Born--Infeld system, etc.} system with a
timelike central singularity, does \emph{not} satisfy the second Bianchi identity
weakly at the center. This is not surprising since the mass function of the RWN
solution goes to $-\infty$ as $r\downarrow 0$ (the conditions \eqref{asymp:m} are of
course only sufficient and not necessary, but we will also rigorously establish that
\eqref{bintA} does not hold).
 
In Section~\ref{sec:Hamiltonians} of this paper we investigate which properties of
the Einstein--Maxwell--Maxwell system are problematic by comparing RWN to spacetimes
of a point charge in different electromagnetic vacua, in particular the Hoffmann
spacetime solution of the Einstein--Maxwell--Born--Infeld (EMBI) system, for which we
show that the second Bianchi identity is satisfied weakly at the singularity. We next
discuss an application of Theorem~\ref{thm:bianchi} to suitable electromagnetic vacua
informally (full details are contained in Section~\ref{sec:Hamiltonians}).

In \cite{TZ1} a particular subclass of admissible\footnote{This particular notion of
admissibility is fully laid out in Section \ref{sec:Hamiltonians}.} electromagnetic
Lagrangians was identified with the property that the corresponding spherically
symmetric, asymptotically flat, electrostatic spacetime metrics have the 
{\em mildest} possible singularity at their center, namely, a conical singularity on
the time axis. In the setting of \cite{TZ1} this is the case only if the bare rest
mass vanishes, i.e., $m_0 = 0$.

In the present work we drop this restriction and allow a nonvanishing bare mass
$m_0$. In fact, since with EIH and Wallace we are interested in a timelike naked
singularity at the center, we need to admit \emph{negative} $m_0$. Note that in the
special-relativistic electrodynamical setting of \cite{KieTah21} the problem is
overdetermined when the bare mass of the particles vanishes, but is well-posed with
nonzero bare mass of either sign. In the general-relativistic setting we expect that
a naked singularity with strictly positive bare mass to be impossible, though (see
Section~\ref{secEMSTmild}). Put differently, we expect that a timelike singularity
with strictly positive bare mass can only exist inside a black hole. Although
timelike singularities with negative bare mass can exist in a black hole, too, they
can also be naked.
 
This generalization opens the door to much more severe than conical, but nevertheless
much weaker singularities than the one at the center of superextremal RWN spacetime.
One key quantity to measure the different degrees of severity of  singularities is 
the Kretschmann scalar. We know that, as $r \to 0$, the Kretschmann scalar is
proportional to $r^{-4}$ in the case of conical singularities studied in \cite{TZ1}.
We will see that it is of order $r^{-6}$ in the case of {\em admissible} reduced
Hamiltonians (as defined in Section~\ref{sec:Hamiltonians}), and that it blows up
like $r^{-8}$ for the RWN solution. Due to the behavior of the cumulative mass
function $m(r)$ at the singularity obtained in Proposition~\ref{lem:mdecay} and
confirming our assumption \eqref{asymp:m}, the restrictions on the reduced
Hamiltonian guarantee that the second Bianchi identity holds weakly everywhere,
including at the singularity.
 
\begin{theorem}\label{thm:es}
Suppose $(\cM,\g,\F)$ is an electrostatic spherically symmetric spacetime for an
admissible reduced Hamiltonian\footnote{Rigorously defined and motivated in
Section~\ref{sec:Hamiltonians}, see Definition~\ref{zeta} on page~\pageref{zeta}.},
and that the bare mass of the central singularity is negative. 
 Then the twice-contracted second Bianchi identity is satisfied weakly.
\end{theorem}

This result is in stark contrast to the RWN solution for which not only the
cumulative mass function (and hence the energy inside a sphere of area $4\pi r^2$)
diverges to minus infinity when $r\downarrow 0$, but even the weak version of the
second Bianchi identity fails. As such, the nonlinear electromagnetic theories
obtained through a Lagrangian formulation are better suited to model static
spacetimes of a charged point particle. We expect that these results can be extended
also to non-symmetric, non-static solutions with several point
charges, or ring singularities, etc.
 
\subsection{Outline}

The rest of this paper is organized as follows.

In Section~\ref{secBIANCHI}, we prove our Main Theorem~\ref{thm:bianchi}. We derive our
sufficient criterion for when the twice-contracted second Bianchi identity holds weakly. 
The assumption of strictly negative bare mass is required for
Corollary~\ref{cor1} of Theorem~\ref{thm:bianchi}, which leads the way to Theorem~\ref{thm:es}.

In Section~\ref{sec:Hamiltonians}, we investigate the Einstein--Maxwell system for a
large family of (nonlinear) electromagnetic vacuum laws. For that we give a precise
formulation of the admissible reduced Hamiltonians and their application to
Theorem~\ref{thm:es}. We will also prove rigorously that the weak second Bianchi
identity does not hold for the well-known RWN metric.

In Section~\ref{summary}, we conclude with a summary and an outlook on possible
extensions of our results. We in particular show by direct computation that the weak
second Bianchi identity holds for a family of singular fluid solutions with vanishing
bare mass, showing that our conditions of our main theorem are not necessary.

 
\section{The weak second Bianchi identity on a spacetime with timelike singularity}
\label{secBIANCHI}

Throughout this section we assume that $\overline{\cM}$ is a manifold (with an
interior boundary $\cT$) that is diffeomorphic to $\RR^4$, equipped with a Lorentzian metric
$\g$ on a part of $\overline{\cM}$ that is diffeomorphic to $\RR \times (\RR^3 {\setminus}\{0\})$.
The part of $\overline{\cM}$ where we are given a Lorentzian metric is denoted by $\cM$. 
In other words, we assume $(\cM,\g)$ to be extendible to a manifold $\overline{\cM}$ with 
an interior one-dimensional timelike boundary $\cT$ that is diffeomorphic to $\RR \times \{0\}$ 
for $0 \in \RR^3$, and assume that $\g$ is a (sufficiently) 
smooth metric tensor on the interior $\cM = \overline{\cM} \setminus \cT$. In our
cases of interest, $(\cM,\g)$ will be generally inextendible to $\overline{\cM}$ as a
sensible Lorentzian manifold (e.g., due to curvature blow-up at a naked singularity)
but our results in this section apply more broadly also to scenarios where $\g$ is 
extendible in some low-regularity fashion (e.g., as continuous metric).

Recall that the classical second Bianchi identity holds on any smooth semi-Riemannian manifold
and thus implies the standard twice-contracted second Bianchi identity~\eqref{Bi2} for the Einstein
tensor pointwise \emph{away} from the singularity $\cT \cong \RR \times \{0\}$. In this section we
derive a distributional Bianchi identity involving the behavior \emph{at} the singularity/boundary 
$\cT$ in a manifold-with-boundary setting following the coordinate blow-up approach 
introduced already in Section~\ref{sec:ZAScoord}. In Theorem~\ref{thm:bianchi} and
Corollaries~\ref{thm:bianchi2} and \ref{cor1} we provide conditions when this weak second Bianchi 
identity, defined in Definition~\ref{def:bianchi}, holds in the static spherically symmetric 
setting of our interest.

\medskip
As discussed in Section~\ref{sec:ZAScoord} we assume that we can write a 
four-dimensional static, spherically symmetric Lorentzian manifold $(\cM,\g)$, which
is diffeomorphic to $\RR \times (\RR^3 \setminus \{ 0 \})$,
using coordinates that are spatially conformally flat.
For $\rho_0>0$ and the new coordinate $\rho \in (\rho_0,\infty)$ we obtain 
\beq\label{gagain}
  ds^2_\g = - e^{2\gamma(\rho)} c^2 dt^2 + \phi^4(\rho)
 [ d\rho^2  +\rho^2 (d\vartheta^2   + \sin^2\vartheta d\varphi^2)].
\eeq
In particular, we view $\cM$ as 
$\RR \times (\RR^3 \setminus \overline{B_{\rho_0}(0)})$ and $\cT$ as 
$\RR \times \p B_{\rho_0}(0)$.

To formulate the second Bianchi identity at the singularity 
$\{r=0\} \cong \{ \rho = \rho_0\}$ in a meaningful way, we first note that {\em if} 
$\mathbf{G}$ were a smooth tensor field defined on the manifold 
$\overline{\cM} \cong \RR \times (\RR^3 \setminus B_{\rho_0}(0))$ \emph{including}
the interior boundary $\cT \cong \RR \times \p B_{\rho_0}(0)$, then \eqref{Bi2},
together with integration by parts and Stokes' Theorem, implies that for any smooth
compactly supported vector field $\psi$ we have
\bna
 0 &=&  \int_{\cM} \psi^\nu \, \nabla_\mu G^{\mu}{}_{\nu} \, d\vol_\g 
    =   \int_{\cM} \nabla_\mu (\psi^\nu G^{\mu}{}_\nu) \, d\vol_\g - \int_{\cM} 
   G^{\mu}{}_\nu \, \nabla_\mu \psi^\nu \, d\vol_\g \nonumber \\
   &=& \int_{\cT} i_{\psi^\nu G^{\mu}{}_\nu} (d\vol_\g) - \int_{\cM} G^{\mu}{}_\nu \, 
   \nabla_\mu \psi^\nu \, d\vol_\g 
   = - \int_{\cM} G^{\mu}{}_\nu \, \nabla_\mu \psi^\nu \, d\vol_\g \label{motivation},
\ena
where the last equality follows from the fact that $\psi$ and $\mathbf{G}$ are
smooth, hence constant on $\partial B_{\rho_0}$, and $\partial B_{\rho_0}$ is a ZAS.
Note that Stokes can be applied here, since the support of $\psi$ is also bounded in $t$.
This motivates us to define the weak formulation of the
twice contracted second Bianchi identity in terms of spacetime integration against
test vector fields.

In more singular situations where $\mathbf{G}$ and thus the interior boundary
integral $\int_\cT$ in \eqref{motivation} may not be well-defined (e.g., in the case
of curvature blow-up) we thus seek an inhomogeneous version of the identity, namely
\[
  \int_\cM G^\mu{}_\nu \nabla_\mu \psi^\nu \, d\vol_\g = \lim_{\varepsilon \to 0} 
  \int_{{\cT}_\varepsilon} i_{\psi^\nu G^{\mu}{}_\nu} (d\vol_\g),
\]
that should furthermore equal zero in the case of a ZAS. Here, 
$(\cT_\varepsilon)_\varepsilon$ with 
$\cT_\varepsilon \cong \RR \times \p B_{\rho_0+\varepsilon}(0)$ is the net converging
to the singularity (compare this to the notation of the Riemannian ZAS 
$\cS_0 \subseteq \p\Sigma$ and its mass in Section~\ref{sec:mass}).

We thus rigorously define a weak version of the second Bianchi identity (for 
$\mathbf{G}$) as follows using an approximation of $\cT$. While in this work we only
focus on the situation where $\cT$ consists of a singular one-dimensional timelike
singularity in the center, it is clear that an analogous definition can be used for
multiple such singularities, or other more general types of interior boundaries.

\begin{definition}\label{def:bianchi}
Let $(\cM,\g)$ be a smooth four-dimensional static, spherically symmetric Lorentzian
manifold, which is diffeomorphic to 
$\RR \times (\RR^3 \setminus \overline{B_{\rho_0}(0)})$ with spatially conformally
flat coordinates $(t,\rho,\vartheta,\varphi)$ (see Section~\ref{sec:ZAScoord}). We say
that the \emph{inhomogeneous twice-contracted second Bianchi identity holds weakly on
$\cM$} if for the Einstein tensor $\mathbf{G} = (G^\mu{}_\nu)$ of $\g$ and for any
compactly supported vector field $\psi \in \mathfrak{X}_c(\overline{\cM})$ the
integral
\beq\label{bianchi1}
 \int_{\cM} G^\mu{}_\nu \nabla_\mu \psi^\nu d\vol_\g
\eeq
exists, and equals
\beq\label{bianchi2}
 \lim_{\varepsilon \to 0} \int_{{\cT}_\varepsilon} i_{\psi^\nu G^\mu{}_\nu} (d\vol_\g),
\eeq
where $\cT_\varepsilon \cong \RR \times \p B_{\rho_0 + \varepsilon}(0)$. We say that
the \emph{twice-contracted second Bianchi identity holds weakly} if 
$\eqref{bianchi1}=\eqref{bianchi2}=0$.
\end{definition}

We can now prove the main result of this paper.

\begin{customthm}{\ref{thm:bianchi}}
Let $\cM \cong \RR \times (\RR^3 \setminus \overline{B_{\rho_0}(0)})$ be equipped
with a static, spherically symmetric Lorentzian metric $\g$ of the form
\eqref{gagain}, i.e., for spatially conformally flat coordinates 
$(t,\rho,\vartheta,\varphi)$ the metric tensor $\g$ is given by
\[
ds_\g^2 = - e^{2\ga(\rho)} c^2 dt^2 + \phi^4(\rho)\left[ d\rho^2 + \rho^2 
(d\vartheta^2 +\sin^2\vartheta d\varphi^2)\right].
\]
Assume furthermore that, as $\rho \downarrow \rho_0$,
\begin{enumerate}
  \item $\phi(\rho) = O(\rho-\rho_0)$, $\phi'(\rho) = O(1)$,
  \item $e^{\gamma(\rho)} = O((\rho-\rho_0)^{-1})$, $\gamma'(\rho) = O((\rho-\rho_0)^{-1})$, 
  and 
  \item $G^\mu{}_\nu(\rho) = O((\rho-\rho_0)^{-5+\kappa})$ for some $\kappa > 0$.
\end{enumerate}
Then
\beq\label{bint}
   \int_\cM G^\mu{}_\nu \nabla_\mu \psi^\nu d\vol_\g = - 4\pi \rho_0^2 c
   \int_{-\infty}^\infty \psi^\rho({t,}\rho_0) \, dt \cdot \lim_{\varepsilon \to 0} 
   G^\rho{}_\rho (\rho_0+\varepsilon) e^{\gamma(\rho_0+\varepsilon)} 
   \phi(\rho_0+\varepsilon)^6 = 0,  
\eeq
for all vector fields $\psi \in \mathfrak{X}_c(\overline{\cM})$ with $\psi(t,\rho_0)$
being independent\footnote{While this restriction is not necessary, it makes sense
because we are just artificially blowing up the one-dimensional singularity 
$\cT \cong \RR \times \{0\}$.} of angular components $(\vartheta,\varphi)$,
and thus the second Bianchi identity is satisfied weakly in the sense of Definition~\ref{def:bianchi}.
\end{customthm}

\begin{proof}
We first show that for any compactly supported vector field 
$\psi \in \mathfrak{X}_c(\overline{\cM})$ the integral
 \beq\label{exi}
  \int_\cM G^\mu{}_\nu \nabla_\mu \psi^\nu \, d\vol_\g
 \eeq
exists. Note that in spatially conformally flat coordinates 
$(t,\rho,\vartheta,\varphi)$, the volume element reads
\[
 d\vol_\g = c e^\gamma \phi^6 \rho^2 \sin \vartheta dt \wedge d\rho \wedge d\vartheta 
 \wedge d\varphi = c e^\gamma \phi^6 dV^4,
\]
where $dV^n$ denotes the Euclidean volume form on $\RR^n$. 
The nonvanishing Christoffel symbols are
\beq \label{Christ} \begin{array}{llll}
 \Ga^{t}_{\rho t} = \gamma' & \Ga^{\rho}_{tt} =  \frac{c^2 e^{2\gamma} \gamma'}{\phi^4} 
 & \Ga^{\rho}_{\rho\rho}= \frac{2\phi'}{\phi} &
 \Ga^{\vartheta}_{\vartheta\rho} =  \Ga^{\varphi}_{\varphi\rho} = \frac{1}{\rho} + 
 \Ga^{\rho}_{\rho\rho} \\
 \Ga^{\rho}_{\vartheta\vartheta}= - \rho^2 \, \Ga^{\vartheta}_{\vartheta\rho} & 
 \Ga^{\rho}_{\varphi\varphi} = \sin^2 \vartheta \, \Ga^{\rho}_{\vartheta\vartheta} & 
 \Ga^{\varphi}_{\varphi\vartheta} = - \cot \vartheta &
 \Ga^{\vartheta}_{\varphi\varphi} = - \sin^2 \vartheta \, \Ga^{\varphi}_{\varphi\rho}.
 \end{array}
\eeq
For proving the existence of the integral \eqref{exi} consider each of the summands
in $\phi^6 G^\mu{}_\nu \nabla_\mu \psi^\nu$ separately. Since $\psi$ is smooth on 
$\overline{\cM}$, all derivatives are bounded, and it remains to consider the contributions of
the Christoffel symbols in $\nabla_\mu \psi^\nu$. Now, 
\[
  G^\mu{}_\nu \nabla_\mu \psi^\nu = G^t{}_t \nabla_t \psi^t + G^\rho{}_\rho \nabla_\rho 
  \psi^\rho + G^\vartheta{}_\vartheta \nabla_\vartheta \psi^\vartheta +G^\varphi{}_\varphi 
  \nabla_\varphi \psi^\varphi
\]
with $\phi^6 G^\mu{}_\nu = O((\rho-\rho_0)^{1+\kappa})$ by assumption (iii), and by
assumptions (i) and (ii) furthermore
\bea
   \nabla_t \psi^t &=& \partial_t \psi^t + \psi^\alpha \Ga^t_{t\alpha} = \partial_t \psi^t 
   + \gamma' \psi^\rho = O((\rho-\rho_0)^{-1}), \\
   \nabla_\rho \psi^\rho &=& \partial_\rho \psi^\rho + \frac{2\phi'}{\phi} \psi^\rho 
   = O((\rho-\rho_0)^{-1}), \\
   \nabla_\vartheta \psi^\vartheta &=& \partial_\vartheta \psi^\vartheta + \left(\frac{1}{\rho} 
   + \frac{2\phi'}{\phi}\right) \psi^\rho = O((\rho-\rho_0)^{-1}), \\
   \nabla_\varphi \psi^\varphi &=& \partial_\varphi \psi^\varphi + \left( \frac{1}{\rho} 
   + \frac{2\phi'}{\phi} \right) \psi^\rho + \cot \vartheta \, \psi^\vartheta = O((\rho-\rho_0)^{-1}).
\eea
Hence $c e^\gamma \phi^6 G^\mu{}_\nu \nabla_\mu \psi^\nu = O((\rho-\rho_0)^{-1+
\kappa})$ and \eqref{exi} exists. 

Next we show that the integral is the same as the limit of the boundary integral on 
$\cT_\varepsilon \subseteq \partial\cM_\varepsilon$ as defined in \eqref{bint}, where
by $\cM_\varepsilon$ we denote the interior that is diffeomorphic to 
$\RR \times (\RR^3 \setminus \overline{B_{\rho_0+\varepsilon} (0)})$. Since the
integral exists,
\[
  \int_\cM G^\mu{}_\nu \nabla_\mu \psi^\nu \, d\vol_\g = \lim_{\varepsilon \to 0} 
  \int_{\cM_\varepsilon} G^\mu{}_\nu \nabla_\mu \psi^\nu \, d\vol_\g.
\]
Because $\g$ is smooth on $\cM_\varepsilon$, the classical second Bianchi identity
immediately implies, as in \eqref{motivation}, that for the vector field 
$X = X^\mu \partial_\mu$, defined by $X^\mu = G^\mu{}_\nu \psi^\nu$, we obtain
by Stokes' Theorem
 \[
  \int_{\cM_\varepsilon} G^\mu{}_\nu \nabla_\mu \psi^\nu \, d\vol_\g = \int_{\cM_\varepsilon} 
  (\ddiv X) \, d\vol_\g = \int_{\partial\cM_\varepsilon} i_{X}(d\vol_\g) 
  = \int_{\cT_\varepsilon} i_{X}(d\vol_\g).
 \]
Here, $i_X (d\vol_\g)$ denotes the interior product of the volume form with $X$.
Recall that only the component of $X$ normal to 
${\cT}_\varepsilon = \RR \times \partial B_{\rho_0+\varepsilon}(0)$ contributes. In
coordinates $(t,\rho,\vartheta,\varphi)$ the outward-pointing\footnote{Recall that 
for a unit normal vector $N$
to $\partial\Omega$ we have $i_X \omega|_{\partial\Omega} = i_{X^\perp}\omega = 
X^0 i_N\omega$ with $X^\perp =\pm g(X,N)N$. Thus, if $N$ is spacelike (and $\partial\Omega$
Lorentzian) it must be chosen outward-pointing and if $N$ is timelike (and $\partial\Omega$
Riemannnian) inward-pointing.} unit normal to 
$\cT_\varepsilon$ is $N = - (g_{\rho\rho})^{-\frac{1}{2}} \partial_\rho$. Since 
$\cT_\varepsilon$ is Lorentzian,
\[
  i_{X} (d\vol_\g)|_{\cT_\varepsilon} = g(X,N)  i_{N} (d\vol_\g) =  g_{\rho\rho} 
  (g_{\rho\rho})^{-\frac{1}{2}} X^\rho \, i_{(g_{\rho\rho})^{-\frac{1}{2}} \partial_\rho} 
  (d\vol_\g) = X^\rho i_{\partial_\rho}(d\vol_\g).
\]
Since $X^\rho = G^\rho{}_\rho \psi^\rho$ does not contain off-diagonal terms, this
implies
\[
    i_{X} (d\vol_\g)|_{\cT_\varepsilon} = G^\rho{}_\rho \psi^\rho i_{\partial_\rho}
    (d\vol_\g),
\]
with
 \[
  i_{\partial_\rho}(d\vol_\g)|_{\cT_\varepsilon} = - c e^{\gamma(\rho_0+\varepsilon)}
  \phi(\rho_0+\varepsilon)^6 (\rho_0+\varepsilon)^2 \sin \vartheta \, dt \wedge d\vartheta
  \wedge d\varphi.
 \]
Hence the integrand of the boundary integral reduces to
 \[
  i_X (d\vol_\g) |_{{\cT}_\varepsilon} = - G^\rho{}_\rho(\rho_0+\varepsilon)
  \psi^{\rho}(t,\rho,\vartheta,\varphi)c e^{\gamma(\rho_0+\varepsilon)} \phi(\rho_0+
  \varepsilon)^6 (\rho_0+\varepsilon)^2 \sin \vartheta \, dt \wedge d\vartheta \wedge
  d\varphi,
 \]
so that
 \bea
  \int_{\cM_\varepsilon} G^\mu{}_\nu \nabla_\mu \psi^\nu \, d\vol_\g 
  = - c G^\rho{}_\rho(\rho_0+\varepsilon) e^{\gamma(\rho_0+\varepsilon)} 
  \phi(\rho_0+\varepsilon)^6 \int_{\cT_\varepsilon} \psi^{\rho} \, dS,
 \eea
where $dS = (\rho_0+\varepsilon)^2 \sin \vartheta \, dt \wedge d\vartheta\wedge d\varphi$.
Since $\psi^\rho$ is smooth (and compactly supported)
 \[
  \lim_{\varepsilon \to 0} \int_{{\cT}_\varepsilon} \psi^\rho \, dS =
  \lim_{\varepsilon \to 0} 4\pi (\rho_0+\varepsilon)^2 \fint_{\partial B_{\rho_0+
  \varepsilon}(0)} \int_{-\infty}^\infty \psi^\rho \, dt \, d\tilde S,
 \]
where $d\tilde S = dS_{\partial B_{\rho_0+\varepsilon}(0)}$ denotes the surface
element of $\partial B_{\rho_0+\varepsilon}(0)$ in flat $\RR^3$. Since $\psi^\rho$
and therefore $\Psi^\rho := \int_{-\infty}^{\infty} \psi^\rho \, dt$ is smooth, we
observe that $\lim_{\varepsilon \to 0} \fint \Psi^\rho \, d\tilde S =
\Psi^\rho(\rho_0)$. We thus obtain that
 \bea
  \int_\cM G^\mu{}_\nu \nabla_\mu \psi^\nu \, d\vol_\g & = & \lim_{\varepsilon \to 0}
  \int_{\cM_\varepsilon} G^\mu{}_\nu \nabla_\mu \psi^\nu \, d\vol_\g \\
  & = & - 4\pi \rho_0^2 c \Psi^\rho(\rho_0) \lim_{\varepsilon \to 0} 
  G^\rho{}_\rho(\rho_0+\varepsilon) e^{\gamma(\rho_0+\varepsilon)} \phi(\rho_0
  +\varepsilon)^6.
 \eea
Due to the assumptions (i)--(iii) it follows that
 \[
  G^\rho{}_\rho(\rho_0+\varepsilon) e^{\gamma(\rho_0+\varepsilon)} \phi(\rho_0
  +\varepsilon)^6 \sim |\rho-\rho_0|^{-5+\kappa-1+6} \to 0 \qquad \text{as }
  {\varepsilon \to 0}.
  \qedhere
 \]
\end{proof}

Even though condition (iii) of Theorem~\ref{thm:bianchi} is close to being optimal for the 
conclusion of that Theorem to hold, the condition may seem somewhat unnatural in view of 
the fact that the components of $G^\mu_\nu$ generally involve both metric coefficients 
$\phi(\rho)$ and $\gamma(\rho)$ together with their first and second derivatives.
We now show that it is possible to eliminate condition (iii) entirely
by strengthening conditions (i) and (ii) to include suitable assumptions on the second derivatives of 
the independent metric coefficients $\phi(\rho)$ and $\gamma(\rho)$.

\begin{corollary}\label{thm:bianchi2}
With the same setup as in Theorem~\ref{thm:bianchi}, assume that the metric coefficients $\phi$ and $\gamma$ satisfy the following stronger assumptions as $\rho_0 \downarrow \rho$:
\begin{itemize}
\item[(i')] Assumptions on $\phi$:
\begin{itemize}
\item[(a)] $\phi(\rho) = O(\rho-\rho_0)$
\item[(b)] $\frac{\phi'}{\phi}(\rho) = (\rho-\rho_0)^{-1} - \frac{\rho+\rho_0}{2\rho\rho_0} + O(\rho-\rho_0)$
\item[(c)] $\frac{\phi''}{\phi}(\rho) = \frac{-2}{\rho} (\rho-\rho_0)^{-1} + O(1)$
\end{itemize}
\item[(ii')] Assumptions on $\gamma$:
\begin{itemize}
\item[(a)] $e^{\gamma(\rho)} = O((\rho-\rho_0)^{-1})$, 
\item[(b)] $\gamma'(\rho) = - (\rho-\rho_0)^{-1} + \frac{1}{2\rho_0} + O(\rho - \rho_0)$,
\item[(c)] $\gamma''(\rho) =  (\rho-\rho_0)^{-2} +O(1)$.
\end{itemize}
\end{itemize}
Then the same conclusion as Theorem~\ref{thm:bianchi} holds, namely, the twice contracted second Bianchi identity is satisfied weakly in the sense of Definition \ref{def:bianchi}.
\end{corollary}

\begin{proof}
A computation shows that for the metric $\mathbf{g}$ of Theorem \ref{thm:bianchi},  the Einstein tensor $\mathbf{G}$ is diagonal, and we have
\beq
 G^\rho{}_\rho= \frac{2}{\rho \phi^4} \left( \gamma' + (1+ \rho \gamma' ) 
 \frac{2\phi'}{\phi} + 2 \rho \left( \frac{\phi'}{\phi}\right)^2\right), \label{GG} 
\eeq
Hence, by (i'b) and (ii'b), $G^\rho{}_\rho = O((\rho-\rho_0)^{-4})$ 
as $\rho \downarrow \rho_0$, while by (i'b) and (i'c) we have
\beq
   G^t{}_t = {\frac{4}{\rho\phi^4} \left( \frac{2 \phi'}{\phi} + \frac{\rho 
   \phi''}{\phi} \right) = O((\rho-\rho_0)^{-4})}. 
\eeq
Finally, the above two, together with (i'b), (ii'b) and (ii'c) imply that
\beq
   G^\vartheta{}_\vartheta= G^\varphi{}_\varphi = \frac{1}{2} (G^t{}_t - G^\rho{}_\rho)
   + \frac{1}{\rho\phi^4} \left( \gamma' \bigg(2+\rho \big(\gamma'+
   \frac{2\phi'}{\phi}\big)\bigg) + \rho \gamma'' \right) = O((\rho-\rho_0)^{-4}),
\eeq
so that hypothesis (iii) of Theorem \ref{thm:bianchi} holds with $\kappa = 1$.
\end{proof}

Sufficient conditions in the usual spherically symmetric coordinates can also be derived:

\begin{corollary}\label{cor1}
Consider $(\cM,\g)$ as in Theorem~\ref{thm:bianchi}. Suppose $\g$ in area-radius
coordinates $(t,r,\vartheta,\varphi)$ is of the form
 \[
  ds^2_\g = - e^{2\alpha(r)} c^2 dt^2 + e^{-2\alpha(r)} dr^2 + r^2 (d\vartheta^2 +
  \sin^2 \vartheta d\varphi^2),
 \]
with
\beq\label{xim}
  e^{2\alpha(r)} = 1 - \frac{2Gm(r)}{c^2r}
\eeq
and such that the cumulative mass function has an absolutely converging power series
expansion
\bna\label{mdecay2}
  m(r)\  {=}\  m_0 + m_1 r + m_2 r^2 + O(r^3),
  \qquad \text{as } r \downarrow 0,
\ena
with $m_0<0$. Then the second Bianchi identity holds weakly in the sense of
Definition~\ref{def:bianchi}.
\end{corollary}

\begin{proof}
By the discussion in Section~\ref{sec:ZAScoord} the transformation to spatially 
conformally flat coordinates $(t,\rho,\vartheta,\varphi)$ is possible. It remains to be
checked that the assumptions (i') and (ii') in Corollary~\ref{thm:bianchi2} are also
satisfied.
 
By assumption~\eqref{mdecay2} we have
\[
  r^2 - \frac{2G}{c^2} r m(r) = O_1(r), \qquad r\downarrow 0,
\]
hence the new coordinate $\rho(r)$ given by \eqref{eq:rhomass} is well-defined for 
small $r>0$ and satisfies
\bea
  \rho(r) &=& \rho_0 \left( 1 + \sqrt{-\frac{2c^2}{Gm_0}} r^{\frac{1}{2}} - \frac{c^2}{Gm_0}
  r - \frac{3c^2 + G m_1}{6 (Gm_0 \sqrt{-\frac{2Gm_0}{c^2}})} r^{\frac{3}{2}} +
  \frac{c^2 m_1}{3Gm_0^2} r^2 + O(r^{\frac{5}{2}}) \right) \\
  &=& \rho_0 + O_2(\sqrt{r}).
\eea
Hence the inverse function $r=r(\rho)$ has an expansion of the form
\bea
  r(\rho) &=& \frac{Gm_0}{2c^2\rho_0^2} ( \rho-\rho_0)^2 \left( -1 + \frac{1}{\rho_0}
  (\rho-\rho_0) + \left(-\frac{1}{4\rho_0^2} + \frac{2Gm_1 - 9 c^2}{12c^2\rho_0^2} 
  \right) (\rho-\rho_0)^2 + O((\rho-\rho_0)^3) \right) \\
  &=&O_2((\rho-\rho_0)^2),
\eea
and thus
\[
  \phi(\rho)^2 = \frac{r(\rho)}{\rho} = O_2(|\rho-\rho_0|^2), \qquad \rho \downarrow
  \rho_0,
\]
with
\beq\label{exp:phi}
  \frac{\phi'}{\phi} = -\frac{1}{2\rho} + \frac{r'(\rho)}{2r(\rho)}
  = (\rho-\rho_0)^{-1} - {\frac{1}{2} \left( \frac{1}{\rho_0} 
  + \frac{1}{\rho} \right)}  + O(\rho-\rho_0),
\eeq
which establishes (i'a)  and (i'b) of Corollary~\ref{thm:bianchi2}.
Moreover, (ii'a) holds since \eqref{mdecay2} implies
\bea
  e^{2\gamma(\rho)} &=& e^{2\alpha(r(\rho))} = 1 - \frac{2 Gm(r(\rho))}{r(\rho)} \\
  &=& 4c^2\rho_0^2(\rho-\rho_0)^{-2} + 4c^2 \rho_0 (\rho-\rho_0)^{-1} + 1- 2Gm_1 + 
  O(\rho-\rho_0) = O_2((\rho-\rho_0)^{-2}),
\eea
and, in particular, using (\ref{exp:phi}),
\bea 
  \gamma'(\rho) &=& \frac{1}{2} (e^{2\gamma(\rho)})'e^{-2\gamma(\rho)} 
= G e^{-2\gamma(\rho)} \frac{r(\rho)'}{r(\rho)}\left( - m'(r(\rho)) +
\frac{m(r(\rho))}{r(\rho)} \right) \\
  &=& {- \frac{1}{2} \left( 1 - \frac{1}{\rho_0} (\rho-\rho_0)+\ldots \right) 
  \left( 2 \frac{\phi'}{\phi} + \frac{1}{\rho} \right) 
\left( 1 + \frac{1}{\rho_0} (\rho-\rho_0) + \ldots \right)} \\
  &=& - \left( \frac{\phi'}{\phi} + \frac{1}{2\rho} \right) \left( 1 + 
  O((\rho-\rho_0)^2) \right)
  = - (\rho-\rho_0)^{-1} + {(2\rho_0)^{-1}} + O(\rho-\rho_0)
\eea
which establishes (ii'b), and upon differentiation, (ii'c) of Corollary~\ref{thm:bianchi2}. 
Finally, using \eqref{exp:phi} we have
 \[
 \frac{\phi''}{\phi} = \left(\frac{\phi'}{\phi}\right)' + \left( \frac{\phi'}{\phi}
 \right)^2  = -(\frac{1}{\rho}+\frac{1}{\rho_0})(\rho-\rho_0)^{-1} + O(1) = \frac{-2}{\rho}(\rho-\rho_0)^{-1} + O(1)
 \]
which establishes (i'c) of Corollary~\ref{thm:bianchi2}, and thus the second
twice-contracted Bianchi identity holds weakly.
\end{proof}

\begin{remark}
The condition that $\beta = -\alpha$ for the coordinate coefficients in
Corollary~\ref{cor1} is not necessary for the coordinate transformation. If 
$\beta \neq - \alpha$ then the cumulative mass function is given as usual by 
$m(r) := \frac{c^2}{2G} r \left( 1 - e^{-2\beta(r)} \right)$ and $m$ needs to define
the same conditions in order for the coordinate transformation to spatially
conformally flat coordinates to go through. However, whether the conditions (ii) and
(iii) in Theorem~\ref{thm:bianchi} hold then not only depends on the behavior of $m$
but also on the asymptotic behavior on $\alpha$ as $r \downarrow 0$ since 
$\gamma(\rho) := \alpha(r(\rho))$.
\end{remark}


\section{Spherically symmetric electrostatic spacetimes}
\label{sec:Hamiltonians}
 
The Einstein--Maxwell equations, i.e., \eqref{EM1} together with
\beq\label{EM2}
 d\F = 0, \qquad d\M = 0,
\eeq
valid pointwise away from any singularities of spacetime, are part of any field
theory of electromagnetism in general relativity. The electromagnetic field is
represented by the Faraday tensor $\F$ and the Maxwell tensor $\M$. To arrive at a
field theory of electromagnetism the tensors $\F$ and $\M$ need to be related by a
``law of the electromagnetic vacuum'' [``ether law'' for short], which also fixes the
corresponding electromagnetic energy-momentum-stress tensor $\T$.

Using Eiesland's Theorem~\cites{E1,E2}, which is a generalized and, in fact,
preceding version of the well-known Birkhoff Theorem, one of us 
\cite{TZ1}*{Theorem 6.2} showed that the metric $\g$ of any electrostatic,
spherically symmetric spacetime with an electromagnetic vacuum law determined by a
field Lagrangian which depends only on the two invariants of $\F$, viz.\ 
$\frac14 F_{\mu\nu}F^{\mu\nu}$ and $\frac14 F_{\mu\nu}\star F^{\mu\nu}$, must, in
spherical coordinates $( t,r,\vartheta,\varphi )$,  be given by
\beq\label{sphsymm}
 ds^2_\g = - e^{2\al(r)} c^2 dt^2 + e^{-2\al(r)} dr^2 + r^2 (d\vartheta^2 + \sin^2
 \vartheta d\varphi^2).
\eeq
The function $\al(r)$ is smooth for $r>0$ and depends on the ether law.
 
The simplest law of the electromagnetic vacuum is Maxwell's ``law of the pure
ether,''
\beq
\label{MV}
\M = - \star\F,
\eeq
where $\star$ is the Hodge star operator\footnote{The Hodge $\star$ dual of a 
$k$-form is a $(n - k)$-form, where $n$ is the number of dimensions. In our setting,
$\star$ takes a $2$-form to the dual 2-form.} with respect to $\g$. In this case the
set of coupled equations (\ref{EM1},\ref{EM2},\ref{MV}) is called the 
Einstein--Maxwell--Maxwell (EMM) system. The unique static, spherically symmetric,
asymptotically flat EMM spacetime is the Reissner--Weyl--Nordstr\"om (RWN) solution
with metric component
\beq\label{RN}
  e^{2\al(r)} = 1 - \frac{2 G}{c^4r} \left(Mc^2- \frac{Q^2}{2r}\right).
\eeq
One can show that $M$ is the ADM mass of the spacetime, while $Q$ is its charge. The
ratio\footnote{In a Newtonian theory, the fraction $\frac{GM_1M_2}{|Q_1Q_2|}$ is the
ratio of the coupling constants of the gravitational and electrical pair interaction
energies of any two interacting point charges.
Inserting empirical values, for two interacting electrons one finds the tiny value 
$\frac{G\mEL^2}{\eEL^2}\approx 2.4\times 10^{-43}$. If one has only one point charge,
it is tempting to think of  $\frac{GM^2}{Q^2}$ as the ratio of the gravitational and
electrical self-energy coupling constants, but in a Newtonian theory there is no such
thing, and in special-relativistic electromagnetic Maxwell--Lorentz field theory of 
point charges, the  self energies are infinite. This does not improve in general
relativity, so the meaning of $\frac{GM^2}{Q^2}$ lies elsewhere.} $\frac{GM^2}{Q^2}$
determines the causal structure of the RWN spacetime ($\frac{GM^2}{Q^2}>1$:
subextremal with black hole region; $\frac{GM^2}{Q^2}=1$: extremal with black hole
region; $\frac{GM^2}{Q^2}<1$: superextremal with a naked singularity). While one
would expect the superextremal RWN spacetime to represent the simplest realistic
charged-particle spacetime, some troubling divergence behaviors occur. More
precisely, the cumulative mass function
\[
m(r) {\ =\ } \frac{c^2}{2G} r\left(1 -  e^{2\al(r)}\right),
\]
(cf. \eqref{cummass}) which for RWN reads
\beq\label{def:massfunRWN}
 m_{\mbox{\tiny{RWN}}}(r)= M- \frac{Q^2}{2c^2r},
\eeq
diverges when $r\downarrow 0$, together with the Kretschmann scalar for RWN (see the
discussion in \cite{TZ1}*{Section 1} and Section~\ref{secEMSTmild} below).

One way to overcome the divergence of the cumulative mass function is to consider a
nonlinear electromagnetic theory, for instance the Born--Infeld theory~\cites{B,BI}
(for more historical context in a modern language, see \cites{K1,K2}). This is done
by choosing a Lorentz- and (Weyl) gauge-invariant Lagrangian density $\bL$ for the 
electromagnetic action
 \[
  \mathcal{S}[\A] = \int_\cM \bL (d\A)
 \]
such that in the weak field limit it reduces to the Lagrangian of the 
Maxwell--Maxwell system (\ref{EM2},\ref{MV}), and has {\em finite} total energy for a
point charge.
 
In the spherically symmetric electrostatic case the above formulation boils down to
finding a suitable {\em reduced Hamiltonian}\footnote{We now largely follow the
notation of \cite{TZ1}, with only smaller deviations.} $\zeta$ that yields a solution
to (\ref{EM1},\ref{EM2}) having an ADM mass $M=\lim_{r\to\infty} m(r) $, and a charge
$Q$. Such a reduced Hamiltonian is subject to the following admissibility criteria.

\begin{definition}\label{zeta}
A function $\zeta\colon \RR^+_0\to\RR$ is called an \emph{admissible reduced
Hamiltonian} if it satisfies
  \begin{itemize}
   \item[(R1)] $\lim_{\mu\to 0} \frac{\zeta(\mu)}{\mu} = 1$.
   \item[(R2)] $\zeta'(\mu)>0$ and $\zeta(\mu) - \mu \zeta'(\mu) \geq 0$ for all $\mu>0$.
   \item[(R3)] $\zeta'(\mu) + 2 \mu \zeta''(\mu) \geq 0$ for all $\mu>0$.
   \item[(R4)] $I_\zeta = 2^{-\frac{11}{4}} \int_0^\infty y^{-\frac{7}{4}} \zeta(y) dy 
   < \infty$.
   \item[(R5)] There exist constants $J_\zeta, K_\zeta, L_\zeta >0$ such that
   \[
    J_\zeta \sqrt{\mu} - K_\zeta \leq \zeta(\mu) \leq J_\zeta \sqrt{\mu}, \qquad \text{and}
    \qquad J_\zeta \sqrt{\mu} - 2 L_\zeta \leq 2 \zeta'(\mu)\mu \leq J_\zeta \sqrt{\mu}.
   \]
  \end{itemize}
\end{definition}

These admissibility criteria are derived and motivated in detail in
Section~\ref{secEMSTmild} below. Note that in terms of the reduced Hamiltonian 
$\zeta$ the cumulative mass function is
\beq \label{mOFr}
m(r) = M - \frac{1}{c^2} \int_r^\infty \zeta \left( \tfrac{Q^2}{2s^4} \right) s^2 ds,
\eeq
and the electric potential is $\A = \varphi(r)c \,dt$ with
\beq
 \varphi(r) = Q \int_r^\infty \zeta' \left( \tfrac{Q^2}{2s^4} \right) \tfrac{1}{s^2}
 ds.
\eeq

One can easily check that the Born--Infeld reduced Hamiltonian
 \[
  \zeta_\mathrm{BI}(\mu) = \sqrt{1+2\mu} - 1
 \]
is an admissible Hamiltonian, and it leads to the Hoffmann solution~\cite{H}. Here, 
$\mu$ is a dimensionless $|D|^2$, where $D$ is Maxwell's displacement field (w.r.t.\
a Lorentz frame).

In Section~\ref{prototype} we give another prototype example with a Zero-Area
singularity (with asymptotic behavior as discussed in Section~\ref{sec:ZAScoord})
which is obtained from a particular admissible reduced Hamiltonian and hence
satisfies the second twice-contracted Bianchi identity weakly.

In Section~\ref{secEMSTmild} we study reduced Hamiltonians $\zeta$ more
systematically. In particular, we discuss the admissibility conditions (R1)--(R5)
from Definition~\ref{zeta} and its consequences. One such important consequence is
that given an electrostatic spacetime solution with charge $Q\in\RR\backslash\{0\}$,
by rescaling the reduced Hamiltonian $\zeta$ we can find another electrostatic
spacetime solution that corresponds to the new, rescaled vacuum law, has the same 
charge $Q$, and has any desired {\em bare rest mass} $m_0\leq 0$ and ADM mass
$M>m_0$. We prove this in Proposition~\ref{prop:beta} in Section~\ref{secEMSTmild}.
Subsequently we revisit the second Bianchi identity and show that an admissible
reduced Hamiltonian implies that the second Bianchi identity is satisfied weakly 
(cf.\ Theorem~\ref{thm:es} in the Introduction) based on our general results in
Section~\ref{secBIANCHI}.

\subsection{A prototype electrovac spacetime with finite negative bare mass}
\label{prototype}

Let $m_0<0$, $M>0$ and $Q\in \RR\setminus\{0\}$ be given, and assume that
\beq
\xi^2 := \frac{G(M-m_0)^2}{Q^2} < 1.
\eeq
Let $(\cM_0,\g_0)$ denote the static, spherically symmetric, asymptotically flat,
electromagnetic spacetime that corresponds to the following reduced Hamiltonian
\beq
\zeta(\mu) = \min\{\mu, \sqrt{\mu_0\mu}\},\qquad \mu_0 := \frac{(M-m_0)^4 c^8}{2Q^6}.
\eeq
It is not hard to see that the mass function of this spacetime will be
\beq\label{mdecay}
m(r) = m_0 + \frac{M-m_0}{2}\left\{\begin{array}{ll} \frac{r}{r_0} & r<r_0\\
2 - \frac{r_0}{r} & r > r_0 \end{array}\right.,
\qquad r_0 := \frac{Q^2}{(M-m_0)c^2}.
\eeq
Thus $m(0) = m_0<0$, $m(\infty) = M>0$. One can also verify that the total charge of 
the spacetime is $Q$, and that the singularity at $r=0$ is not shielded by a horizon.

We observe that on this 3-parameter family of electrovac spacetimes, indexed by 
$m_0, M, Q$, the asymptotic behavior \eqref{mdecay} together with
Corollary~\ref{cor1} immediately implies that the twice-contracted second Bianchi
identity is satisfied weakly.

\subsection{Nonlinear electrostatic spacetimes with naked singularities} 
\label{secEMSTmild}

We defined admissible reduced Hamiltonians in Definition~\ref{zeta} via the
properties (R1)--(R5).
The reason for these requirements are discussed extensively in \cite{TZ1}*{Sec.\ 4}.
Let us briefly mention that (R1) implies that the weak field limit is the same as for
classical linear electromagnetics, (R2) implies the dominant energy condition is
satisfied, and (R3) guarantees that $\zeta$ is the Legendre--Fenchel transform of a 
Lagrangian density. Note that (R3) together with (R1) implies the strong energy
condition.

If $m$ denotes the cumulative mass function~\eqref{mOFr}, then 
$\lim_{r\to\infty}m(r)= M$, and the metric components in coordinates 
$(t,r,\vartheta,\varphi)$ in \eqref{sphsymm} are given by 
$e^{2\al(r)}= 1- \frac{2G m(r)}{c^2r}$ as in \eqref{xim}.

If $\lim_{r\downarrow0}m(r) =: m_0$ is finite and nonzero, the Kretschmann scalar
\beq\label{def:Kretch}
 \cK := R_{\mu\nu\lambda\eta} R^{\mu\nu\lambda\eta} =
 \frac{4G^2}{c^4 r^6} \left( 12 m^2 + 4rm (-4m'+rm'') + r^2(8m'^2-rm'm''+r^2 m''^2) \right)
\eeq
blows up at least like $r^{-6}$ as $r \downarrow 0$, indicating the presence of a
true singularity as opposed to a mere coordinate singularity. In the case of the
negative-mass Schwarzschild metric the Kretschmann scalar $\cK = \frac{48 G^2 M^2}{c^4 r^6}$ 
is moreover proportional to the square of the mass $M$. For the 
Reissner--Weyl--Nordstr\"om metric \eqref{RN} with ADM mass $M$ and charge $Q$ (the
Maxwell law $\zeta(\mu)=\mu$ satisfies the criteria (R1)--(R3) but not more), the
Kretschmann scalar 
\beq\label{rn1}
  \cK_\mathrm{RWN}(r)
  = \frac{48G^2M^2}{c^4 r^6} \left( 1 + \frac{2Q^2}{c^2 Mr} + 7 
  \frac{Q^4}{48c^4 M^2r^2} \right)
\eeq
blows up like $r^{-8}$ as $r \downarrow 0$.

Thus the above conditions alone already imply that the center at $r=0$ must be
nonregular. In \cite{TZ1} the mildest possible naked singularity was studied and 
found to be a conical singularity with $m_0=0$ and blow-up rate $r^{-4}$ of the
Kretschmann scalar. More generally \emph{naked} singularities require $m_0 \leq 0$;
otherwise black holes will occur. In other words, there are two cases to consider:

\medskip
{\bf Case 1: $m_0 > 0$.}
We show that, so long as $M < \infty$, the spacetime will have a horizon. In this
vein, assume to the contrary that there is no horizon, so that the coordinate chart 
$(t,r,\vartheta,\varphi)\in \RR\times\RR_+\times\SS^2$ is global, with $r$ spacelike and
$t$ timelike throughout. Set 
\[
f(r) := c^2 r - 2Gm(r).
\]
Then, since $m_0>0$ by assumption, $f(0) = -2Gm_0 < 0$, while $M<\infty$ clearly
implies $f(\infty)>0$. Since $m(r)$ is a continuous function for $r>0$, it now
follows that there exists an $r_0>0$ such that $f(r_0) = 0$. But this implies that
$e^{2\al(r_0)} = \frac{1}{c^2}\frac{f(r_0)}{r_0}= 0$. In fact the metric coefficient
$g_{00}$ generally changes sign across $r=r_0$, which means that the area-radius
coordinate $r$ becomes timelike for $r<r_0$, in contradiction to the hypothesis that
$r$ is spacelike throughout. Therefore $r=r_0$ is a Killing horizon, and the
coordinate chart only covers the region $r>r_0$ of the spacetime.

\medskip
{\bf Case 2: $m_0 \leq 0$.}
In the superextremal RWN spacetime one has $\lim_{r\downarrow0}m(r) =-\infty$ and a
severe curvature singularity at $r=0$. This is due to a non-integrable electrostatic
field energy density about $r=0$. However, for admissible Hamiltonians with finite
electrostatic field energy, one can compute the total field energy to be
\[
\int_0^\infty \zeta \left( \tfrac{Q^2}{2s^4} \right) s^2 ds = |Q|^{\frac{3}{2}}
I_\zeta,
\qquad
  I_\zeta := 2^{-\frac{11}{4}} \int_0^\infty y^{-\frac{7}{4}} \zeta(y) \, dy 
  {\ < \infty},
\]
as demanded in (R4), and we therefore have 
$M - m_0 = \frac{|Q|^{\frac{3}{2}} I_\zeta}{c^2}$; i.e., the difference between the 
accumulated mass of the spacetime and the bare rest mass is entirely due to the
electrostatic field. For the ADM mass $M$ we then have, in general,
\beq\label{adm}
 M = m_0+\frac{|Q|^{\frac{3}{2}} I_\zeta}{c^2}.
\eeq
Thus we have $\lim_{r\downarrow0}m(r) =: m_0 \in (-\infty, 0]$, yielding a less 
singular behavior at $r=0$.

\bigskip
From now on we assume that $m_0 \leq 0$ is finite and consider the behavior of the
spacetime near the center. If we assume that there exists a positive constant
$J_\zeta$ such that
\begin{enumerate}
 \item[(R5${}^\prime$)] $\zeta(\mu) \leq J_\zeta \sqrt{\mu}$,
\end{enumerate}
then the integral term of the  cumulative mass function $m(r)$ in \eqref{mOFr} can be
estimated using
\[
 \int_0^r \zeta\big(\tfrac{Q^2}{2s^4}\big) s^2 ds \leq J_\zeta |Q| 
 \frac{r}{\sqrt{2}},
\]
which implies that in a neighborhood of the center, $m(r)$ is bounded from above by
\beq\label{mup}
 m(r) \leq m_0 + J_\zeta |Q| \frac{r}{\sqrt{2} c^2}. 
\eeq
Since by assumption $m_0\leq 0$, this shows that the metric coefficient $g_{00}$ is
bounded away from zero, 
 \[
e^{2\al(r)} = 1 - \frac{2G m(r)}{c^2r} \geq 1 - \frac{2Gm_0}{c^2r} - \sqrt{2} 
J_\zeta |Q| \frac{G}{c^4} > 1 - \sqrt{2} J_\zeta |Q| \frac{G}{c^4} > 0,
 \]
as long as
\beq\label{nakedq0}
  \frac{|Q|G J_\zeta}{c^4} < \frac{1}{\sqrt{2}}.
\eeq
(Note that the left-hand-side is dimensionless.) Thus, \eqref{nakedq0} implies the
absence of a horizon, which means that a \emph{naked singularity} occurs whenever the
charge is sufficiently small. Of course, \eqref{nakedq0} is only a sufficient
condition for absence of a horizon. 
 
\medskip
We now show that if an electrostatic spacetime solution exists for prescribed total
charge $Q$, then for the same charge $Q$ one can generate such a spacetime with any
bare rest mass $m_0\leq 0$ and ADM mass $M>m_0$. This is achieved by a rescaling of
the associated reduced Hamiltonian.
 
\begin{proposition}\label{prop:beta}
Let $\zeta$ be an admissible Hamiltonian that satisfies \emph{(R1)--(R3)}. We
additionally assume that $\zeta$ satisfies
\begin{enumerate}
 \item[(R4)] $I_\zeta = 2^{-\frac{11}{4}} \int_0^\infty y^{-\frac{7}{4}} \zeta(y) dy 
 < \infty$.
\end{enumerate}
Suppose there exists an electrostatic spacetime metric $\g$ with charge 
$Q\in \RR\backslash\{0\}$ satisfying the Einstein--Maxwell equations for the ether
law generated by $\zeta$. Let $m_0\leq 0$ and $M > m_0$ be given. Then for the 
dimensionless number
\[
 \lambda {\ :=\ } \frac{|Q|^{\frac{3}{2}} I_\zeta}{(M-m_0)c^2}
\]
the $\lambda$-scaled version of $\zeta$, defined by
\beq\label{zetab}
  \zeta_\lambda (\mu)= \lambda^{-4} \zeta (\lambda^4 \mu),
\eeq
is itself an admissible reduced Hamiltonian, and there exists a corresponding
electrostatic spacetime metric $\g_\lambda$ which has charge $Q$, ADM mass
$M=\lim_{r\to\infty}m(r)$, and {\em bare rest mass }
\beq\label{def:brm}
m_0 := \lim_{r\downarrow 0} m(r).
\eeq
\end{proposition}

\begin{proof}
Note that $\zeta_\lambda$ satisfies (R1)--(R4) because $\zeta$ does. Furthermore, 
\eqref{zetab} implies that $I_{\zeta_\lambda}$ as defined in (R4) transforms as
\beq
 I_{\zeta_\lambda} = 2^{-\frac{11}{4}} \int_0^\infty y^{-\frac{7}{4}} \lambda^{-4}
 \zeta(\lambda^4 y) dy = \lambda^{-1} I_\zeta.
\eeq
Therefore, by \eqref{adm},
 \[
  M = m_0 + \frac{|Q|^{\frac{3}{2}} I_{\zeta_\lambda}}{c^2}
 \]
as desired.
\end{proof} 
 
\begin{remark}
The borderline case $m_0 = 0$  was treated already in \cite{TZ1}. In this case there
is no bare mass at the center, and the geometric ADM mass $M$ is entirely due to the
electrostatic field energy $|Q|^{\frac{3}{2}} I_\zeta$, more precisely,
$
M = \frac{|Q|^{\frac{3}{2}} I_\zeta}{c^2}.
$
In \cite{TZ1} it was also shown that given any charge $Q$, also \emph{any} positive
ADM mass $M>0$ can be achieved in this case via an appropriate choice of a scaling
parameter: For $\zeta_\lambda(\mu) = \lambda^{-4} \zeta (\lambda^4\mu)$ we have
$I_{\zeta_\lambda} = \lambda^{-1} I_\zeta$, and the ADM mass becomes
$
M = \frac{1}{ \lambda}\frac{|Q|^{\frac{3}{2}} I_\zeta}{c^2},
$
with $Q$ still the charge of the spacetime. By a suitable choice of $\lambda$, any
value of $M>0$ can be generated. Clearly, this is a special case of our Proposition
\ref{prop:beta}. These solutions are asymptotically flat with a conical singularity
at the center if the ratio $\frac{GM^2}{Q^2}$ is sufficiently small; see 
\cite{TZ1}*{Sec.\ 5.1}.
\end{remark}
 
The sufficient condition \eqref{nakedq0} for obtaining a naked singularity can be
reformulated in the $\lambda$-scaled setting. Note that
$
 J_{\zeta_\lambda} = \lambda^{-2} J_\zeta.
$
Hence \eqref{nakedq0} translates to
\beq\label{eps}
  \frac{G (M-m_0)^2 }{|Q|^2} < \frac{I_\zeta^2}{\sqrt{2} J_\zeta}.
\eeq
From now on we always assume that \eqref{eps} is satisfied.

Together with $m_0 \leq 0$, condition \eqref{eps} guarantees that there is no horizon
and $r$ is a spacelike coordinate on $(0,\infty)$. In fact, we have
 
\begin{proposition}
Suppose $\zeta$ is an admissible reduced Hamiltonian satisfying 
\emph{(R1)--(R5$^\prime$)} and $\lambda$ etc.\ is given as in Proposition~\ref{prop:beta}. 
If the dimensionless ratio
\[
\epsilon^2 := \frac{G(M-m_0)^2}{|Q|^2}
\]
is sufficiently small (as in \eqref{eps}), then $\g_\lambda$ features a naked
singularity at the center. \hfill \qed
\end{proposition}

Note that $\epsilon^2$ is a {\em dimensionless} quantity in Gaussian units (see also
\cite{M}*{p.\ 5} for a discussion).

\begin{example}[Born--Infeld model]
In the setting of the Born--Infeld theory, where $\zeta(\mu) = \sqrt{1+2\mu}-1$, we
can choose $J_\zeta = \sqrt{2}$ in (R5${}^\prime$). Moreover,
\[
 I_\zeta = - \frac{\Gamma(-\frac{3}{4})\Gamma(\frac{5}{4})}{2\sqrt{\pi}} \approx
 1.2360498,
\]
and thus
\[
 \frac{I_\zeta^2}{\sqrt{2} J_\zeta} \approx 0.76390954.
\]
If we consider the mass and charge of an electron, i.e.,
\[
 M_\text{e} = 9.10938356  \times 10^{-31} ~\text{[kg]}, \qquad Q_\text{e} 
 = 1.6021765 \times 10^{-19} ~\text{[C]} \cdot \sqrt{k_\text{e}},
\]
where $k_\text{e} = 8.98755179 \times 10^9 ~\text{[kg} \,\text{m}^3
\text{s}^{-2}\text{C}^{-2}]$ is the Coulomb constant, then for $m_0 = 0$, and 
gravitational constant $G = 6.67408 \times 10^{-11} [\text{m}^3 \text{kg}^{-1}
\text{s}^{-2}]$, we have
\[
 \epsilon^2 = \frac{G M^2}{|Q|^2} \approx 2.40053 \times 10^{-43},
\]
so we are far in the naked singularity regime due to \eqref{eps} being satisfied.
Since gravitational effects ($\propto G$) are small, for $m_0<0$ we are guaranteed a
naked singularity so long as $m_0 > - \varsigma M$, where $\varsigma$
is a large positive
constant.
\end{example}
 
Next, let us consider the behavior of the spacetime near the center of the symmetry,
for $m_0<0$. The singularity at $r=0$ will no longer be conical, but exhibit a
stronger blow-up behavior. If, in addition to (R5${}^\prime$) we assume that there is
also $J_\zeta, K_\zeta >0$ such that
\begin{enumerate}
 \item[(R5${}^\pp$)] $J_\zeta \sqrt{\mu} - K_\zeta \leq \zeta(\mu)$, 
\end{enumerate}
then we also obtain an estimate of $m(r)$ from below. More precisely,
\beq\label{mlow}
 m(r)=m_0+\frac{1}{c^2}\int_0^r\zeta \Big(\tfrac{Q^2}{2s^4} \Big) s^2 ds \geq m_0 +
 J_\zeta |Q| \frac{r}{\sqrt{2}c^2} - K_\zeta \frac{r^3}{3c^2},
\eeq
which together with {the upper bound} \eqref{mup} implies that
\[
 m(r) = m_0 + J_\zeta |Q| \frac{r}{\sqrt{2}c^2} + O(r^3)
\]
as $r \downarrow 0$.

If we in addition assume that there is a positive constant $L_\zeta > 0$ such that
\begin{enumerate}
 \item[(R5${}^\ppp$)] $J_\zeta \sqrt{\mu} - 2 L_\zeta \leq 2 \zeta'(\mu)\mu \leq
 J_\zeta \sqrt{\mu}$,
\end{enumerate}
then we can also {infer} something about the decay of the derivatives of $m(r)$. We
combine all properties (R5${}^\prime$)--(R5${}^\ppp$) in (R5).
 
\begin{proposition}\label{lem:mdecay}
If $\zeta \colon \RR^+_0 \to \RR$ is an admissible reduced Hamiltonian, that is, it
satisfies the properties \emph{(R1)--(R4)} as well as
\begin{enumerate}
   \item[(R5)] There exist positive constants $J_\zeta, K_\zeta, L_\zeta$ such that
   \[
    J_\zeta \sqrt{\mu} - K_\zeta \leq \zeta(\mu) \leq J_\zeta \sqrt{\mu}, \qquad
    \text{and} \qquad J_\zeta \sqrt{\mu} - 2 L_\zeta \leq 2 \zeta'(\mu)\mu \leq
    J_\zeta \sqrt{\mu},
   \]
\end{enumerate}
then the cumulative mass function is of the form
\[
  m(r) = m_0 + \frac{1}{2c^2}\int_0^r \zeta \Big( \frac{Q^2}{2s^4} \Big) s^2 ds 
  = m_0 + \frac{J_\zeta |Q|}{\sqrt{2} c^2} r - \frac{K_\zeta}{3 c^2} r^3 + O_2(r^3),
  \qquad \text{as} ~ r \downarrow 0,
\]
where we say that $f(r) = O_k(r^\alpha)$ as $r\downarrow 0$ if 
$r^{j-\alpha} \frac{d^j f}{dr^j}$ is bounded for $j=0,\ldots, k$ as $r\downarrow 0$.
\end{proposition}

\begin{remark}
By the first part of (R2), in particular, $\zeta'(\mu)  \geq 0$ and thus also 
$\zeta(\mu) \geq 0$ by the second part for all $\mu \geq 0$. Next, consider 
$f(\mu) = \log \frac{\zeta(\mu)}{\mu}$. Then $f(0) = \log 1 = 0$ by (R1) and by (R2)
\[
  f'(\mu) = \frac{\mu \zeta'(\mu) - \zeta(\mu)}{\mu \zeta(\mu)} \leq 0.
\]
Hence by integration also $f(\mu) \leq 0$ and therefore $\zeta(\mu) \leq \mu$ for all
$\mu \geq 0$.  Together with the first part of (R5) we thus obtain for $\mu \geq 0$ 
that
\beq\label{zeta1}
  \max\{0,J_\zeta \sqrt{\mu} - K_\zeta\} \leq \zeta(\mu) \leq 
  \min\{\mu,J_\zeta \sqrt{\mu}\}.
\eeq
Similarly, (R2) implies that $0 \leq \zeta'(\mu) \mu \leq \zeta(\mu)$ so that 
together with the second part of (R5) we have for $\mu \geq 0$ that
\beq\label{zeta2}
  \max\{0,J_\zeta \sqrt{\mu} - 2 L_\zeta\} \leq 2 \zeta'(\mu)\mu \leq 
  \min\{2\mu, J_\zeta \sqrt{\mu}\}.
\eeq
We will use these inequalities in our proof of Proposition~\ref{lem:mdecay} below.
\end{remark}

\begin{proof}
As we have already seen in \eqref{mup}--\eqref{mlow} the first part of (R5) implies 
that
\[
  m_0 + \frac{J_\zeta |Q|}{\sqrt{2}c^2} r - \frac{K_\zeta}{3c^2} r^3 \leq  m(r) 
  \leq m_0 + \frac{J_\zeta |Q|}{\sqrt{2}c^2} r,
\]
and thus shows that
\[
0 \leq r^{-3} \left[ m(r) - \left( m_0 + \frac{J_\zeta |Q|}{\sqrt{2}c^2} r 
- \frac{K_\zeta}{3c^2} r^3 \right) \right] \leq \frac{K_\zeta}{3c^2}, 
\]
remains bounded. Since
$
  m'(r) = \zeta\big( \tfrac{Q^2}{2r^4 c^2} \big) r^2,
$
using the first part of (R5) we again obtain that
$
  \frac{J_\zeta |Q|}{\sqrt{2} c^2} - \frac{K_\zeta r^2}{c^2} \leq m'(r) 
  \leq \frac{J_\zeta |Q|}{\sqrt{2} c^2},
$
hence
\[
  0 \leq r^{-2} \frac{d}{dr} \left[ m(r) - \left( m_0 + 
  \frac{J_\zeta |Q|}{\sqrt{2} c^2} r - \frac{K_\zeta}{3 c^2} r^3 \right) \right] 
  \leq \frac{K_\zeta}{c^2}.
\]
The second derivative of $m(r)$ satisfies
\[
 - \frac{2r K_\zeta}{c^2} \leq m''(r) = \frac{2r}{c^2} \zeta \Big( \frac{Q^2}{2r^4} \Big) - 4 \zeta' 
  \Big( \frac{Q^2}{2r^4} \Big) \frac{Q^2}{2 r^{4}} \frac{r}{c^2}
  \leq \frac{\sqrt{2} J_\zeta |Q|}{r c^2} - \frac{\sqrt{2} J_\zeta |Q|}{r c^2}
  + \frac{4 L_\zeta r}{c^2} = \frac{4 L_\zeta r}{c^2},
\]
and thus
\[
  0 \leq r^{-1} \frac{d^2}{dr^2} \left[ m(r) - \left( m_0 + 
  \frac{J_\zeta |Q|}{\sqrt{2} c^2} r - \frac{K_\zeta}{3 c^2} r^3 \right) \right] 
  = r^{-1} \left[ m''(r)+ \frac{2K_\zeta}{c^2} r \right] \leq \frac{4L_\zeta 
  + 2K_\zeta}{c^2}
\]
is bounded as well. Therefore, by definition of $O_2(r^3)$,
\[
 m(r) = m_0 + \frac{J_\zeta |Q|}{\sqrt{2}c^2} r - \frac{K_\zeta}{3 c^2} r^3 
 + O_2(r^3), \qquad \text{as } r \downarrow 0. \qedhere
\]
\end{proof}
 
With the results obtained for admissible nonlinear theories we are now in a position
to show that the weak second Bianchi identity does hold for spherically symmetric 
electrostatic spacetimes where the reduced Hamiltonian $\zeta$ satisfies (R1)--(R5).

\begin{customthm}{\ref{thm:es}}
Suppose $(\cM,\g,\F)$ is a electrostatic spherically symmetric spacetime considered
in Proposition~\ref{lem:mdecay} such that $m_0<0$ and there is a naked singularity at
the center. Then the twice-contracted second Bianchi identity holds weakly on $\cM$.
\end{customthm}
 
\begin{proof}
By Proposition~\ref{lem:mdecay}, $m(r) = m_0 + \frac{J_\zeta |Q|}{\sqrt{2} c^2} r 
- \frac{K_\zeta}{3 c^2} r^3 + O_2(r^3), \text{ as } ~ r \downarrow 0$. Hence the
second Bianchi identity is satisfied weakly also at the singularity due to
Corollary~\ref{cor1}.
\end{proof}
 
\begin{remark}
Note that even though the value of $m_0$ is not relevant for whether the Bianchi
identity holds or does not hold weakly, its sign does matter, since we use the radial
variable $r$ all the way down to $r=0$, which is not possible in the presence of a
horizon.
\end{remark}

\subsection{The RWN spacetime does not satisfy the weak second Bianchi identity}

In the previous subsection we identified a whole class of electrostatic spacetimes
for which the second Bianchi identity {\em does} hold weakly. Using spatially
conformally flat coordinates we {now} show that the RWN spacetime \emph{does not} 
satisfy the weak second Bianchi identity. Since this coordinate transformation is 
rather involved in practice, we also include a heuristic explanation of this ``too 
singular'' behavior of the RWN spacetime in terms of the blow-up rate of the
Kretschmann scalar.

We recall that the RWN metric is of the form
\[
   ds^2_\g = - e^{2\alpha(r)} c^2 dt^2 + e^{-2\alpha(r)} dr^2 + r^2 (d\vartheta^2 
   + \sin^2 \vartheta \, d\varphi^2), 
\]
where
 \[
  e^{2\alpha(r)} = 1 - \frac{2GM}{c^2 r} + \frac{GQ^2}{c^4 r^2} =: 1 - 
  \frac{2}{r} A + \frac{1}{r^2} B^2
 \]
with ADM mass $M$ and charge $Q$. We are interested in the superextremal case 
$\frac{A^2}{B^2} = \frac{GM^2}{Q^2} < 1$ which has a naked singularity in the center,
and will show that in this case the second Bianchi identity does not hold weakly.
 
Note that the cumulative mass function
 \[
  m(r) = M - \frac{Q^2}{2c^2 r}
 \]
blows up at the center and hence does \emph{not} have the required asymptotics 
\eqref{asymp:m} discussed earlier. However, the integral appearing in the definition
of $\rho$ in \eqref{eq:rhomass}, i.e.,
 \[
  \int_0^r \frac{dr'}{\sqrt{r'^2 - \frac{2G}{c^2} r'm(r')}} = \int_0^r 
  \frac{dr'}{\sqrt{r'^2 - 2Ar'+B^2}}
 \]
does converge for $r<B = \frac{\sqrt{G}|Q|}{c^2}$ since the denominator satisfies
 \[
  0 < (r'-B)^2 < r'^2 - 2Ar' + B^2 < (r'-A)^2
 \]
due to superextremality. Hence a coordinate transformation to spatially conformally
flat coordinates $(t,\rho,\vartheta,\varphi)$ is possible near the singularity.
 
If we set
 \[
  \rho_0 := B-A = \frac{\sqrt{G}|Q|-G M}{c^2}>0
 \]
then the change of coordinates \eqref{eq:rhomass} is given by (compare to the 
calculation of the negative mass Schwarzschild case in Section~\ref{sec:mass})
 \bea
  \rho(r) = (r-A) + \sqrt{(r-A)^2 + (B^2 - A^2)}, \qquad r 
  =  \frac{(\rho+A)^2-B^2}{2\rho}.
 \eea
The conformal factor, $\phi$, is then
 \[
  \phi(\rho)^2 = \frac{r}{\rho} = \frac{(\rho+A)^2-B^2}{2\rho^2} = O(\rho-\rho_0), 
 \]
and 
 \[
  e^{2\gamma(\rho)} = e^{2\alpha(r(\rho))} = 1 - \frac{2}{\rho\phi^2} A 
  + \frac{1}{\rho^2\phi^4} B^2 
  = \frac{\rho^2 \phi^4 - 2\rho\phi A + B^2}{\rho^2 \phi^4} = O((\rho-\rho_0)^{-2}).
 \]
where the blow-up rate follows from the above since $A < B$. Note that some of these
decay/blow-up rates already differ from the requirement in Theorem~\ref{thm:bianchi}
(i) and (ii). Without looking into all the requirements of Theorem~\ref{thm:bianchi}
separately, we directly jump to investigate the critical quantity 
$G^\rho{}_\rho e^\gamma \phi^6$ appearing in \eqref{bint}. To this end we observe 
that
\bea
  \gamma' &=& \frac{(e^{2\gamma})'}{2e^{2\gamma}} = (e^{-2\gamma}-1) \left( 
  \frac{1}{\rho}+2 \frac{\phi'}{\phi} \right) = O((\rho-\rho_0)^{-1}),
\eea
since 
\bea
   \frac{\phi'}{\phi} = \frac{(\phi^2)'}{2\phi^2} 
   = \frac{(B^2 - A^2) - 2A\rho}{2\rho ((\rho+A)^2 - B^2)} = O((\rho-\rho_0)^{-1}).
\eea
Thus by \eqref{GG}
\bea
  G^\rho {}_\rho
 &=& \frac{1}{\rho^2\phi^4 - 2\rho\phi A + B^2} \left(2 + 8\rho\frac{\phi'}{\phi} 
 + 8\rho^2\left( \frac{\phi'}{\phi} \right)^2 \right)
- \frac{1}{\rho^2\phi^4} - \frac{4}{\phi^4} \left( \frac{\phi'}{\phi} \right)^2,
 \eea
hence
$
  G^\rho{}_\rho = O((\rho-\rho_0)^{-4})  
$
and 
$
  G^\rho{}_\rho e^\gamma \phi^6 = O((\rho-\rho_0)^{-2}).
$
Therefore, the second term in \eqref{bint} does not converge to zero, and the second
Bianchi identity does not hold weakly. In fact, even the inhomogeneous second Bianchi
identity does not hold weakly.


\section{Summary and Outlook}
\label{summary}

In this paper we have considered the following question:Under what conditions is the
twice-contracted second Bianchi identity satisfied in a weak sense in a neighborhood
of a singular line of a spacetime $\cM$ with the metric $\g$?
We were able to answer this question in case $(\cM,\g)$ is both static and spherically 
symmetric, by finding sufficient conditions on the metric that, if satisfied, guarantee 
that the weak second Bianchi identity holds everywhere, the location of a timelike 
singularity included. 

The main application of this result is to electrovacuum spacetimes with timelike
singularities. We have shown that the Einstein--Maxwell equations, complemented with
a nonlinear vacuum law which satisfies certain admissibility conditions, have
spherically symmetric, static solutions describing the electrostatic spacetime of a
point charge with weakly satisfied twice-contracted second Bianchi identity. We also
found that the Bianchi identity is not weakly satisfied by the 
Reissner--Weyl--Nordstr\"om (RWN) solution, which is obtained when complementing the
Einstein--Maxwell equations with the standard linear vacuum law of Maxwell. The
favorable electrostatic spacetimes turn out to be less singular than RWN, a fact 
that is evident from the blow-up behavior of their curvature invariants. In our setting,
for example, the Kretschmann scalar of an electrostatic spacetime with weakly
satisfied twice-contracted second Bianchi identity blows up at most like $r^{-6}$ as
$r \downarrow 0$, while it blows up as $r^{-8}$ in the RWN solution. In the case of a
vanishing bare rest mass, i.e., $m_0=0$, the blow-up rate is only $r^{-4}$ as 
$r \downarrow 0$, leading to the mildest possible (a conical) singularity.

Our findings add another argument to the many that have already been offered for why
Maxwell's linear vacuum law \eqref{MV} should be replaced by a nonlinear law that 
reduces to \eqref{MV} in the weak-field limit (for in this limit the Maxwell--Maxwell
electromagnetic field equations are indisputably successful), see \cite{Mie}, and
which furnishes finite field energies of point charges (unlike Maxwell's law of the 
vacuum), see \cites{B,BI} for the most prominent earliest voices in this regard. The
family of possible laws that allow for a weak twice-contracted second Bianchi
identity is huge, so that one has to look elsewhere for arguments that could help
narrowing down the list of potential candidates. The Born--Infeld law \cite{BI}
stands out in this regard because it follows from a handful of compelling principles,
see \cites{B,P}, each of which seemingly capturing some aspect of nature. 
Since all these models depend on at
least one extra parameter, and reduce to Maxwell's law in the limit where this
parameter vanishes, experimental results can restrict the realm of possible parameter
values for each model, and possibly rule out specific models, but it is difficult to 
see how empirical results alone could hint at the ``right'' nonlinear model. To find 
the right one --- if this is indeed the message --- one needs to argue based on
deeper plausible principles, as Born and Infeld tried, and Plebanski \cite{P} and
Boillat \cite{Boi1970} did.

The conditions that yield a weak second Bianchi identity, which we derived in
Theorem~\ref{thm:bianchi} and several corollaries, can be applied more generally, to
static spherically symmetric Lorentzian manifolds with a singularity in the center.
We note that the proof of our Corollary~\ref{cor1} does, in general, \emph{not}
extend to the case $m_0=0$, because the coordinate transformation from $r$ to $\rho$
would involve the infinite integral $\int_0^r \frac{dr'}{r'}$. However, we believe a 
result similar to Corollary~\ref{cor1} can be obtained as long as $m_0 \leq 0$ and
$m(r)=O(r^\kappa)$, $0\leq \kappa < 1$, as $r\downarrow 0$. Moreover, there are cases
with zero bare mass $m_0=0$ and conical singularities that do admit a transformation
to spatially conformally flat coordinates \eqref{gagain} if we allow $\rho_0=0$ and 
possibly interpret the Bray mass using the more general definition of 
\cite{BJ}*{Section 3.2} via approximation of regular ones. Examples are an explicit
singular solution of the astrophysically important Tolman--Oppenheimer--Volkoff
equation, studied already by Chandrasekhar~\cite{Chand} and others~\cites{AB,Mak} in
view of its asymptotics and discussed below, and also the Hoffmann spacetime
discussed in \cites{H,TZ1}.

\begin{example}[A singular static spherically symmetric fluid with vanishing bare 
mass that satisfies the weak second Bianchi identity]
For solutions of the Einstein--Euler equations with linear equation of state
$p=(kc)^2\mu$, $k\in(0,1)$, with conical singularity described in 
\cite{AB}*{Section 3.2.1} the cumulative mass function is of the form
\[
  { m(r) = \frac{Kc^2}{2G} r,\qquad K:=\frac{4k^2}{4k^2+(1+k^2)^2}, \qquad 
  r \in (0,\infty),}
\]
and hence goes to zero when $r\downarrow 0$. Thus $e^{-2\beta} = 1 - \frac{2Gm}{c^2} 
= 1 - K <1$. Hence the transformation to spatially conformally flat coordinates
\eqref{gagain} is given by $\rho=r^{\frac{1}{\sqrt{1-K}}}$ and 
$\phi^2 = \rho^{\sqrt{1-K}-1}$ if we allow $\rho_0=0$. In this case the weak version
of the second Bianchi identity can be obtained directly in the $(t,r,\vartheta,\phi)$
coordinates since $r^2 e^{\beta(r)} dr = \rho^2 {\phi^6} d\rho$ and 
$\alpha(r) = \gamma(\rho)$ etc. In particular, the limit in \eqref{bint} translates
to
\bea
 \lim_{\varepsilon \to 0} \varepsilon^2 G^r{}_r(\varepsilon)e^{\alpha(\varepsilon)
 +\beta(\varepsilon)} 
 = \lim_{\varepsilon\to0} e^{\alpha(\varepsilon)-\beta(\varepsilon)} 
 \left(1 - e^{2\beta(\varepsilon)} + \frac{4k^2}{1+k^2}\right) 
 = C_k \lim_{\varepsilon\to 0}\varepsilon^{\frac{2k^2}{1+k^2}} = 0,
\eea
where $C_k$ is a constant depending on $k$. This ends our example.
\end{example}

Furthermore, it is clear that the main ideas developed in this paper are not
restricted to static spherically symmetric spacetimes and are adaptable to more
general situations. In particular, we expect that our results can be extended to 
non-static, non-symmetric spacetimes with finitely many timelike singularities, which
appear as point-type, or perhaps (st)ring-type singularities in the spacelike leaves
of any foliation of the spacetime into ``evolving spaces''; for a study of the 
equation of motion of singularities of the latter type, see \cite{AD}. We expect that
the less severe blow-up behavior demanded by the weak Bianchi identity will point the 
way to the formulation of a well-posed dynamical theory for charged timelike
singularities and the electromagnetic spacetime structures around them. By requiring
compatible singularities in the electromagnetic energy-momentum-stress tensor this in
turn should lead to the identification of an admissible class of electromagnetic
vacuum laws. Our preliminary inquiry in this direction also indicates that the
admixture of a scalar field that modulates the gravitational coupling of the
electromagnetic field energy-momentum-stress tensor to the spacetime curvature may be
needed. Moreover, due to the occurrence of off-diagonal components in the Einstein 
tensor, it is reasonable to expect that more restrictions on the metric may be
required in order to obtain a broadly applicable result analogous to 
Theorem~\ref{thm:bianchi}.

In all these cases we also expect the bare mass of the singularity to be strictly
negative. We recall that it had to be strictly negative in the spacetimes studied in
the present paper for the weak twice-contracted second Bianchi identity to hold with
a rigorous geometric interpretation of the bare rest mass as Bray's ZAS mass.
Moreover, in the general-relativistic setting we expect positive bare mass to imply a
black hole (cf. the discussion in \cite{GW}), not a naked singularity, and as such
could not serve as a suitable point-charge model of physical ``particles'' like 
nuclei, or electrons.

To this we add the following thought: Given the spectacular high precision agreement 
of quantum-mechanical computations of atomic spectra with the empirical ones,
modifications of Maxwell's vacuum law would have to be significant only in the
immediate vicinity of a point charge. This in turn suggests that the electromagnetic
self-field energies, though finite, will still be huge, and this in turn implies that
the bare mass would have to be negative, so that the total mass agrees with the
empirical mass as obtained in scattering experiments. This is an argument for why
even in special relativity a consistent electrodynamical theory of fields and their
point charge sources that also agrees with observations would have to be formulated
with negative bare mass. Our finding that we were able to establish the weak second
twice-contracted Bianchi identity only for naked singularities with negative bare
mass seems fitting.


\end{document}